%% file: main.tex
\documentclass[onefignum,onetabnum]{siamonline190516}   

\input{ex_shared}

\ifpdf
\hypersetup{
  pdftitle={Ensemble Kalman Filters with Resampling},
  pdfauthor={Omar Al-Ghattas, Jiajun Bao, and Daniel Sanz-Alonso}
}
\fi

\begin{document}

\maketitle

\begin{abstract}
Filtering is concerned with online estimation of the state of a dynamical system from partial and noisy observations. In applications where the state of the system is high dimensional, ensemble Kalman filters are often the method of choice. These algorithms rely on an ensemble of interacting particles to sequentially estimate the state as new observations become available. Despite the practical success of ensemble Kalman filters, theoretical understanding is hindered by the intricate dependence structure of the interacting particles. This paper investigates ensemble Kalman filters that incorporate an additional resampling step to break the dependency between particles. 
The new algorithm is amenable to a theoretical analysis that extends and improves upon those available for filters without resampling, while also performing well in numerical examples.
\end{abstract}

\begin{keywords}
Ensemble Kalman filter, resampling, non-asymptotic error bounds
\end{keywords}

\begin{AMS}
    62F15, 68Q25, 60G35, 62M05
\end{AMS}

\section{Introduction}\label{sec:introduction} 
The filtering problem of estimating a time-evolving state from partial and noisy observations arises in numerous applications, including numerical weather prediction, automatic control, robotics, signal processing, machine learning, and finance \cite{sarkka2023bayesian,crisan2011oxford,reich2015probabilistic,asch2016data,law2015data,majda2012filtering,sanzstuarttaeb2023}. 
When the state is high dimensional and the dynamics governing its evolution are complex, the method of choice is often the ensemble Kalman filter ($\mathsf{EnKF}$) \cite{evensen1994sequential,evensen1996assimilation,evensen2009data,houtekamer2016review,evensen2022data}. 
In this filtering algorithm, a Kalman gain matrix defined via the first two moments of an ensemble of particles determines the relative importance given to the dynamics and the observations in estimating the state.
The size of the ensemble controls both the accuracy and the computational cost of the algorithm. Operational implementations of $\mathsf{EnKF}$ give accurate state estimation with a moderate ensemble size, significantly smaller than the state dimension \cite{houtekamer2016review}. However, non-asymptotic theory that explains the successful performance of $\mathsf{EnKF}$ with moderate ensemble size is still not fully developed. An important impediment to such a theory is the presence of correlations between particles, since the Kalman gain used to update each particle depends on the entire ensemble. 
This paper investigates a modification of $\mathsf{EnKF}$ that incorporates a resampling step to break these correlations. The new algorithm
is amenable to a theoretical analysis that extends and improves upon those available for filters without resampling, while also maintaining a similar empirical performance.

\subsection{Resampling in Filtering Algorithms}
Resampling techniques are routinely employed to enhance particle filtering algorithms which assimilate observations by weighting particles according to their likelihood \cite{del2004feynman,doucet2009tutorial}. For particle filters, resampling converts weighted particles into unweighted ones to alleviate weight degeneracy and achieve variance reduction at later times \cite[Chapter 9]{chopin2020introduction}. In contrast, $\mathsf{EnKF}$ assimilates observations by using \emph{unweighted} particles and relying on a Gaussian \emph{ansatz} and Kalman-type formulae.  $\mathsf{EnKF}$ avoids weight degeneracy by design, but remains vulnerable to filter divergence and ensemble collapse \cite{harlim2010catastrophic,kelly2015concrete}; several works have proposed using resampling to remedy these issues.

An early discussion of resampling for $\mathsf{EnKF}$ can be found in \cite{anderson1999monte}, which replaces the standard Gaussian \emph{ansatz} with a more flexible sum of Gaussian kernels. The paper \cite{zhang2010improving} introduced bootstrap methods for identifying and alleviating the impact of spurious correlations, thereby enhancing the robustness of the Kalman gain. The work \cite{lawson2004implications} proposed a resampling scheme to improve the performance of deterministic filters in nonlinear settings. This method involves periodically resampling the ensemble based on a ``bootstrapping'' approach as suggested by \cite{anderson1999monte}, which is fundamentally based on a kernel density technique taken from the particle filtering literature. Closest to our work is the paper \cite{myrseth2013resampling}, which demonstrates that resampling the Kalman gain in the conditioning step of $\mathsf{EnKF}$ can help prevent the ensemble from collapsing over time, consequently enhancing ensemble stability and reliability. The numerical experiments in \cite{myrseth2013resampling} suggest that relative to the non-resampled setting, $\mathsf{EnKF}$ algorithms that employ resampling give more reliable prediction intervals with a slight trade-off in the accuracy of their point predictions. 

Ensemble Kalman methods are also used for offline parameter estimation and, relatedly, as numerical solvers for inverse problems, see e.g. \cite{gu2007iterative,aanonsen2009ensemble,li2007iterative,iglesias2013ensemble,chada2021iterative}. While not the focus of this paper, we point out that resampling techniques have also been investigated in this context. For instance,
\cite{wu2022resampled} removes particles that significantly deviate from the posterior distribution via a resampling procedure, thus improving the performance of standard implementations. A similar idea is also considered in 
\cite{wu2019improving}, which proposes adding an extra resampling step in each iterative cycle. This method improves the convergence of the iterative $\mathsf{EnKF}$ by perturbing the shrinking ensemble covariances to prevent early stopping while preserving the consistent Kalman update direction of standard implementations.

\subsection{Our Contributions}
Whereas previous work investigates resampling from a methodological viewpoint \cite{anderson1999monte,zhang2010improving,lawson2004implications,myrseth2013resampling}, the primary objective of this paper is to demonstrate that resampling strategies provide a promising approach to the design of ensemble Kalman algorithms with non-asymptotic theoretical guarantees. We consider a simple parametric resampling scheme: at the beginning of each filtering step, members of the ensemble are independently sampled from a Gaussian distribution whose mean and covariance match those of the ensemble at the previous time-step. Thereafter, the filtering step can be carried out using any of the numerous existing $\mathsf{EnKF}$ variants \cite{evensen2009data,tippett2003ensemble}. For the resulting algorithm, which we term $\mathsf{REnKF}$, we establish theoretical guarantees that extend and improve upon those available for filters without resampling. 

Our theoretical guarantees hold in the linear-Gaussian setting in which we provide a detailed error analysis of the ensemble mean and covariance as estimators of the mean and covariance of the filtering distributions, given by the Kalman filter \cite{kalman1960new}. Our theory covers both stochastic and deterministic dynamical systems; in addition, it covers both stochastic implementations based on \emph{perturbed observations} \cite{evensen2003ensemble} and deterministic implementations based on \emph{square-root} filters \cite{tippett2003ensemble,anderson2001ensemble,bishop2001adaptive}. Importantly, our error-bounds are \emph{non-asymptotic} and \emph{dimension-free}: they hold for any given ensemble size and are written in terms of the \emph{effective-dimension} of the covariance of the initial distribution, and of the dynamics and observation models. The non-asymptotic and dimension-free analysis of ensemble Kalman updates has recently been considered in \cite{ghattas2022non}, which demonstrated rigorously the success of ensemble Kalman updates whenever the ensemble size scaled with the effective dimension of the state as opposed to its ambient dimension. Given that ensemble Kalman algorithms are often employed in problems where the state dimension is very large, our results also contribute to the theoretical understanding of why ensemble methods are able to perform well even when the ensemble size is taken to be much smaller than the state dimension. This paper extends the results in \cite{ghattas2022non} by providing new bounds over multiple assimilation cycles. Our work may also be compared to \cite{majda2018performance}, which puts forward a non-asymptotic and dimension-free analysis of a multi-step $\mathsf{EnKF}$ that utilizes a different modification than the one used to define $\mathsf{REnKF}$. Specifically, \cite{majda2018performance} employs an additional projection step that determines the effective dimension of the method.

Other multi-step analyses were limited to square-root filters with deterministic dynamics \cite{kwiatkowski2015convergence,ghattas2022non} and to \emph{asymptotic} analysis of stochastic implementations \cite{kwiatkowski2015convergence}, which, while ensuring consistency of the filters, do not explain their practical success when deployed with a small ensemble size. The key reason why existing non-asymptotic analyses \cite{ghattas2022non} do not extend to stochastic implementations and dynamics is that these additional sources of randomness further complicate the correlations between particles, which we break via resampling. 

We numerically illustrate the theory in a linear setting and also demonstrate the successful performance of $\mathsf{REnKF}$ on the Lorenz 96 equations \cite{lorenz1996predictability}, a simplified model for atmospheric dynamics widely used to test filtering algorithms \cite{majda2006nonlinear,majda2012filtering,law2016filter,sanz2015long}. In our experiments, $\mathsf{REnKF}$ performs similarly to standard, non-resampled $\mathsf{EnKF}$ in fully and partially-observed settings. Moreover, the results are robust to the noise level in the dynamics and in the observations. Python code to reproduce all numerical experiments is publicly available at \href{https://github.com/Jiajun-Bao/EnKF-with-Resampling}{https://github.com/Jiajun-Bao/EnKF-with-Resampling}.

\subsection{Outline} The rest of the paper is organized as follows. Section \ref{sec:formulation} formalizes the problem setting and provides necessary background on $\mathsf{EnKF}$. Section \ref{sec:algorithms} introduces and analyzes the new $\mathsf{REnKF}$ algorithm. The main result, Theorem \ref{thm:MultiStepPOEnKFBounds}, gives non-asymptotic and dimension-free error bounds. We report numerical results that confirm and complement the theory in Section \ref{sec:numerics}. Proofs are collected in Section \ref{sec:proofmaintheorem}. We close in Section \ref{sec:Conclusions} with a discussion of our results and directions for future research.

\subsection{Notation} For a vector $u = \bigl(u(1),\dots, u(d)\bigr)^\top$ and $q \ge 1$, $|u|_q = (\sum_{i=1}^d |u(i)|^q)^{1/q}$ and $|u| = |u|_2$. For a random variable $X$ and $q \ge 1$, we write $\normn{X}_q = (\E|X|^q)^{1/q}$ and $\normn{X}=\normn{X}_2$. $X \sim \Nc(m, C)$ denotes that $X$ is a Gaussian random vector with mean $m$ and covariance $C$, and we denote its density at a point $x$ by $\Nc(x; m, C)$. $\mcS^d_+$ denotes the set of $d\times d$ symmetric positive-semidefinite matrices, and $\mcS^d_{++}$ denotes the set of $d\times d$ symmetric positive-definite matrices. For two $d\times d$ matrices $A,B$, $A \succ B$ implies $A-B \in \mcS^d_{++}$ and $A \succeq B$ implies $A-B \in \mcS^d_{+}$, and similarly for $\prec, \preceq$. For a $n \times m$ matrix $A = (A_{ij})^{n, m}_{i=1,j=1}$, the operator norm is given by $|A| = \sup_{\normn{v}_2=1} |Av|_2$. $\indicator \{S\}$ denotes the indicator of the set $S$. The identity matrix will be denoted by $I,$ and on occasion its dimension will be made explicit with a subscript. The $n\times m$ zero matrix will be denoted by $O_{n \times m}$.

\section{Problem Setting and Ensemble Kalman Filters}\label{sec:formulation}
We consider a $d$-dimensional unobserved \emph{state process} $\{u^{(j)}\}_{j \ge 0}$ and a $k$-dimensional \emph{observation process} $\{y^{(j)}\}_{j \ge 1}$ whose relationship over discrete time $j$ is governed by the following hidden Markov model:
\begin{alignat}{4}
&\text{(Initialization)} &  \qquad     u^{(0)} &\sim \Nc( \mu^{(0)}, \Sigma^{(0)}), \label{eq:HMMinitialization}\\
&\text{(Dynamics)}& \qquad         u^{(j)} &= \Psi(u^{(j-1)}) + \xi^{(j)},  \quad \quad  &\xi^{(j)} \iid \Nc(0, \Xi), \qquad j = 1, 2, \ldots \label{eq:HMMdynamics} \\
&\text{(Observation)}  & \qquad      y^{(j)} &= H u^{(j)} + \eta^{(j)},  \quad \quad  &\eta^{(j)} \iid \Nc(0, \Gamma),     \qquad j = 1, 2, \ldots \label{eq:HMMobservations}
\end{alignat}
We assume that the initial distribution $\Nc( \mu^{(0)}, \Sigma^{(0)}),$ where $\mu^{(0)}\in \R^d$, $\Cpost^{(0)}\in \mcS^d_{++}$, the model dynamics map $\Psi : \R^d \to \R^d,$ the observation matrix $H \in \R^{k \times d},$ and the dynamics and observation noise covariance matrices $\Xi \in \mcS^d_+, \Gamma \in \mcS^k_{++}$ are known; otherwise, they may be estimated from the observations, see e.g. \cite{evensen2022data,chen2022autodifferentiable,chen2023reduced}. We further assume that the random variables $u^{(0)},$  \(\{\xi^{(j)}\}_{j \ge 1}\), and  \(\{\eta^{(j)}\}_{j\ge 1}\) are mutually independent. 
All methods and theory presented in this paper extend immediately to dynamics and/or observation models that are not time homogeneous at the expense of a more cumbersome notation. Additionally, nonlinear observations can be dealt with by augmenting the state, see e.g. \cite{anderson2001ensemble}.

For a given time index $j \in \N,$ the \textit{filtering} goal is to compute the \emph{filtering distribution} $p \bigl(u^{(j)} | Y^{(j)}\bigr)$, where  $Y^{(j)}:= \{y^{(1)}, \ldots,y^{(j)}\}.$ The filtering distribution provides a probabilistic summary of the state  $u^{(j)}$ conditional on observations up to time $j$.
Given access to the filtering distribution at the preceding time-step $j-1$, $p \bigl(u^{(j)} | Y^{(j)}\bigr)$ may be obtained by the following two-step procedure:
\begin{alignat}{4}
  &\text{(Forecast)} \quad \quad   &p \bigl(u^{(j)} | Y^{(j-1)}\bigr) &&&= \int \Nc (u^{(j)}; \Psi( u^{(j-1)}), \Xi) p\bigl(u^{(j-1)} |Y^{(j-1)} \bigr) \, d u^{(j-1)},     \label{eq:TrueForecastStep}               \\
  &\text{(Analysis)}  &p \bigl(u^{(j)} | Y^{(j)}\bigr) &&& \propto  \Nc( y^{(j)}; H u^{(j)}, \Gamma) p \bigl(u^{(j)} | Y^{(j-1)} \bigr)\label{eq:TrueAnalysisStep}.
\end{alignat}
The \emph{forecast distribution} $p \bigl(u^{(j)} | Y^{(j-1)}\bigr)$ represents our knowledge of the state at time $j$ given observations up to time $j-1$, and its computation in \eqref{eq:TrueForecastStep} utilizes the dynamics model \eqref{eq:HMMdynamics}. In the analysis step \eqref{eq:TrueAnalysisStep}, the new observation $y_{j}$ is assimilated through an application of Bayes formula with prior given by the forecast distribution and likelihood determined by the observation model \eqref{eq:HMMobservations}. Closed-form expressions for the filtering and forecast distributions are only available for a small class of hidden Markov models \cite{papaspiliopoulos2014optimal}. For problems outside this class, many algorithms have been developed to approximate the filtering distributions, or, if this is too costly, to find point estimates of the state \cite{sarkka2023bayesian,sanzstuarttaeb2023}.   

This paper is concerned with $\mathsf{EnKF}$ algorithms that belong to the larger family of Kalman methods. These methods invoke a Gaussian \emph{ansatz} for the forecast distribution, so that Bayes formula in the analysis step can be readily applied using the conjugacy of the Gaussian forecast distribution and the Gaussian likelihood model \eqref{eq:HMMobservations}.
The distinctive feature of $\mathsf{EnKF}$ is that the Gaussian approximation is defined using the first two moments of an ensemble of particles. Then, in the analysis step each individual particle is updated with a Kalman gain matrix which incorporates the forecast covariance. Several stochastic and deterministic implementations for the analysis step have been proposed in the literature, see e.g. \cite{houtekamer2016review,tippett2003ensemble,evensen2009data}. In Algorithm~\ref{algEnKF}, an example of a stochastic implementation of $\mathsf{EnKF}$ ---commonly referred to as the Perturbed Observation $\mathsf{EnKF}$--- is provided for reference, and will be our focus for this work. 
At time $j=0,$ an \textit{initial ensemble} of $N$ particles are independently drawn from the initial distribution in \eqref{eq:HMMinitialization}. These ensemble members are then sequentially passed through forecast and analysis steps: In the forecast step, the ensemble is propagated through the system dynamics yielding the $j$-th \textit{forecast ensemble}. In the analysis step, the new observation $y^{(j)}$ is assimilated by updating each ensemble member according to a Kalman-type formula, yielding the $j$-th \textit{analysis ensemble}. Although the initial ensemble members are mutually independent, the dependence structure of the ensemble is highly non-trivial beginning at the analysis step at time $j=1$. Indeed, note that the Kalman Gain $K^{(1)}$ is a nonlinear transformation of the entire forecast ensemble, and this matrix is used to update each of the ensemble members when constructing the analysis ensemble. The recursive nature of the algorithm further complicates the dependence structure of the ensemble, rendering a non-asymptotic analysis highly challenging. 

The stochastic variant of $\mathsf{EnKF}$ in Algorithm~\ref{algEnKF} is arguably the most popular in applications \cite{evensen1994sequential,van2020consistent}. Unfortunately, as noted in \cite{furrer2007estimation, ghattas2022non} and further discussed in Section~\ref{sec:algorithms}, it is harder to analyze from a non-asymptotic viewpoint than deterministic variants of the $\mathsf{EnKF}.$

\begin{algorithm} 
\caption{\label{algEnKF} Ensemble Kalman Filter ($\mathsf{EnKF}$)}
\begin{algorithmic}[1]

\STATE {\bf Input}:  $\Psi, H, \Xi, \Gamma, \mu^{(0)}, \Sigma^{(0)}, N.$  Sequentially acquired data $\{y^{(j)}\}_{j\ge 1}.$
\STATE {\bf Initialization}: 
    $u_n^{(0)} \iid \Nc(\mu^{(0)},\Sigma^{(0)}), \quad 1\le n \le N.$
\STATE For $j = 1, 2, \ldots$ do the following forecast and analysis steps:
\STATE {\bf Forecast}: 
\begin{align}\label{eq:predEnKF}
\begin{split}
\widehat{u}_{n}^{(j)} &= \Psi(u_{n}^{(j-1)})+\xi^{(j)}_{n}, \quad \xi_n^{(j)} \iid \Nc(0,\Xi), \quad  1 \le n \le N, \\
\widehat{m}^{(j)} &= \frac{1}{\Sam}\sum^{\Sam}_{\sam=1} \widehat{u}_{n}^{(j)}, \qquad \widehat{C}^{(j)} = \frac{1}{\Sam-1}\sum^{\Sam}_{\sam=1}\bigl(\widehat{u}^{(j)}_{n}-\widehat{m}_{j}\bigr) \bigl(\widehat{u}^{(j)}_{n}-\widehat{m}_{j}\bigr)^\top.
\end{split}
\end{align}
\vspace{-0.5cm}
\STATE{{{\bf Analysis}}}: 
	\begin{align}\label{eq:analEnKF}
	\begin{split}
 K^{(j)} &= \widehat{C}^{(j)} H^\top (H \widehat{C}^{(j)} H^\top + \Gamma)^{-1}, \\
y^{(j)}_{\sam}&=y^{(j)}+ \eta^{(j)}_{n},\quad  \eta^{(j)}_{n} \iid \Nc(0,\Gamma), \quad  1 \le n \le N,\\
u_{n}^{(j)} &= (I-K^{(j)}H)\widehat{u}^{(j)}_{n}+K^{(j)}y_n^{(j)},\quad 1 \le \sam \le \Sam,  \\
\widehat{\mu}^{(j)} &= \frac{1}{N} \sum_{n=1}^N u_{n}^{(j)}, \qquad 
\widehat{\Sigma}^{(j)} = \frac{1}{N-1} \sum_{n=1}^N (u_{n}^{(j)} - \widehat{\mu}^{(j)})(u_{n}^{(j)} - \widehat{\mu}^{(j)})^\top.
\end{split}
\end{align}
\vspace{-0.5cm}
\STATE{\bf Output}:  Analysis mean $\widehat{\mu}^{(j)}$ and covariance $\widehat{\Sigma}^{(j)}$  for $j =1, 2, \ldots$  
\end{algorithmic}
\end{algorithm}

The output $\widehat{\mu}^{(j)}$ of $\mathsf{EnKF}$ gives a point estimate of the state $u^{(j)}$ at time $j$. For such a state-estimation task, $\mathsf{EnKF}$ is very effective \cite{law2012evaluating}.  Additionally, the output $\widehat{\Sigma}^{(j)}$ may be used to construct confidence intervals. However, as often noted in the literature \cite{ernst2015analysis,law2012evaluating} and further discussed in Section \ref{sec:numerics}, caution should be exercised when using ensemble Kalman algorithms for such uncertainty quantification tasks. 
$\mathsf{EnKF}$ performance for state estimation and uncertainty quantification tasks can be assessed by the error in approximating the mean and covariance of the filtering distributions; the theory in Subsection \ref{sec:theory} adopts such performance metrics. If the moments of the filtering distributions are not available, performance metrics such as root mean squared error and coverage of confidence intervals can be employed \cite{law2012evaluating}, and we do so in the numerical experiments in Section \ref{sec:numerics}. 

\section{Ensemble Kalman Filters with Resampling}\label{sec:algorithms}
In this section, we first introduce and motivate our main algorithm, $\mathsf{EnKF}$ with resampling ($\mathsf{REnKF}$). We then present the non-asymptotic theoretical analysis of $\mathsf{REnKF}$ in a linear model dynamics setting. 

\subsection{Main Algorithm}
The idea underlying $\mathsf{REnKF},$ which is outlined in Algorithm~\ref{algEnKFresample}, is to employ a resampling step at each filtering cycle to break the correlations between ensemble members described in Section~\ref{sec:formulation}. We consider here a particularly simple parametric resampling scheme in which at the beginning of each filtering cycle, ensembles are independently sampled from a Gaussian distribution whose mean and covariance match those of the analysis ensemble at the previous time step. Although the resampling mechanism can be made to be more sophisticated ---for example, one may consider nonparametric resampling schemes in which the empirical distribution of the ensemble is used instead--- we note that such complications may be difficult to justify given the simplicity, theoretical guarantees (Subsection \ref{sec:theory}), as well as the computational scalability and empirical performance (Section~\ref{sec:numerics}) of the proposed resampling strategy. Other than the resampling step, the forecast and analysis steps of $\mathsf{REnKF}$ agree with those of $\mathsf{EnKF},$ and consequently any of the stochastic or deterministic implementations of $\mathsf{EnKF}$ can be adopted. Our focus here is on the stochastic implementation of $\mathsf{EnKF}$ in Algorithm \ref{algEnKF}. As discussed in the next subsection ---see Remarks \ref{rem:DeterministicImplementations} and \ref{as}--- non-asymptotic theory for deterministic implementations can be obtained as a by-product of the theory that we develop. 

\begin{algorithm} 
\caption{\label{algEnKFresample} Ensemble Kalman Filter with Resampling ($\mathsf{REnKF}$)}
\begin{algorithmic}[1]
\STATE {\bf Input}: $\Psi, H, \Xi, \Gamma, \mu^{(0)}, \Sigma^{(0)}, N.$ Sequentially acquired data $\{y^{(j)}\}_{j\ge 1}.$
\STATE {\bf Initialization}: Set $\widehat{\mu}^{(0)} = \mu^{(0)}$ and
    $\widehat{\Sigma}^{(0)} = \Sigma^{(0)}.$
\STATE For $j = 1, 2, \ldots$ do the following resampling, forecast, and analysis steps:
 \STATE{{{\bf Resampling}}}: 
    \begin{align}
    \begin{split}
  {u}_{n}^{(j-1)}  &\iid \Nc( \widehat{\mu}^{(j-1)}, \widehat{\Sigma}^{(j-1)}), \quad 1\le n \le N.
    \end{split}
    \end{align}
\vspace{-0.5cm}
\STATE {\bf Forecast}: Do \eqref{eq:predEnKF}.
\STATE{{{\bf Analysis}}}: Do \eqref{eq:analEnKF}.
\STATE{\bf Output}: Analysis mean $\widehat{\mu}^{(j)}$ and covariance $\widehat{\Sigma}^{(j)}$  for $j =1, 2, \ldots$  
\end{algorithmic}
\end{algorithm}

Notice from Algorithm~\ref{algEnKFresample} that correlations between particles could alternately be broken by resampling between the forecast and analysis steps. While such an approach would be amenable to a non-asymptotic analysis akin to the one we develop, we empirically found that resampling after the forecast step significantly deteriorates the performance of the filter in nonlinear settings.
A heuristic explanation is that resampling tacitly introduces a Gaussian approximation, and the filtering distribution is better approximated by a Gaussian than the forecast distribution when the dynamics are nonlinear and the observations are Gaussian. 

\subsection{Non-asymptotic Error Bounds}\label{sec:theory}
Here we present theoretical guarantees for $\mathsf{REnKF}$ in a linear dynamics setting. We introduce the setting and necessary background in Subsection \ref{ssec:settingtheory}. Then, the main result is stated and discussed in Subsection \ref{ssec:mainresult}.

\subsubsection{Setting and Preliminaries}\label{ssec:settingtheory}
We consider $\mathsf{REnKF}$ in the following linear version of the hidden Markov model governing the relationship between the state and observation processes:
\begin{alignat}{4}
&\text{(Initialization)} &  \qquad     u^{(0)} &\sim \Nc( \mu^{(0)}, \Sigma^{(0)}), \label{eq:linearinit}\\
&\text{(Dynamics)}& \qquad         u^{(j)} &= Au^{(j-1)} + \xi^{(j)},  \quad \quad  &\xi^{(j)} \iid \Nc(0, \Xi), \qquad j = 1, 2, \ldots  \label{eq:lineardynamics} \\
&\text{(Observation)}  & \qquad      y^{(j)} &= H u^{(j)} + \eta^{(j)},  \quad \quad  &\eta^{(j)} \iid \Nc(0, \Gamma),     \qquad j = 1, 2, \ldots  \label{eq:linearobservations}
\end{alignat}
with \(u^{(0)}\) independent of the i.i.d. sequences \(\{\xi^{(j)}\}\) and  \(\{\eta^{(j)}\}\). Thus, we assume that the dynamics map $\Psi$ in \eqref{eq:HMMdynamics} is linear and represented by a given matrix $A \in \R^{d \times d}.$ In this case, it is well known that the forecast distributions $p\bigl( u^{(j)} | Y^{(j-1)} \bigr) = \Nc( u^{(j)} ; m^{(j)}, C^{(j-1)})$ and the filtering distributions $p\bigl( u^{(j)} | Y^{(j)} \bigr) = \Nc( u^{(j)} ; \mpost^{(j)}, \Cpost^{(j)})$ are both Gaussian, and the means and covariances of these distributions are given by the Kalman filter \cite{sanzstuarttaeb2023}. We aim to derive non-asymptotic bounds between the output $\widehat{\mu}^{(j)}$ and  $\widehat{\Sigma}^{(j)}$ of $\mathsf{REnKF}$ and the output $\mpost^{(j)}$ and $\Cpost^{(j)}$ of the Kalman filter.

We follow the exposition in \cite{kwiatkowski2015convergence} and introduce three operators that are central to the theory: the \emph{Kalman gain} operator $\msK$, the \emph{mean-update} operator $\msM,$ and the \emph{covariance-update} operator $\msC,$ defined respectively by
\begin{align}
  \msK: \mcS_+^d \to \R^{d \times k}, 
  \qquad & \msK(\Cpr) = \Cpr \H^\top (\H \Cpr \H^\top + \Gamma)^{-1}, \label{eq:KalmanGainOperator}\\
  \msM: \R^d \times \mcS_+^d \to \R^{d},
  \qquad &  \msM(\mpr, \Cpr; y)  = \mpr + \msK(\Cpr) (y-\H\mpr),\label{eq:MeanOperator}\\
  \msC: \mcS_+^d \to \mcS_+^d,
  \qquad &
   \msC(\Cpr) = \bigl(I-\msK(\Cpr)\H \bigr)\Cpr. \label{eq:CovarianceOperator}
\end{align}
With this notation, the mean and covariance updates from time $j-1$ to time $j$ given by the Kalman filter are summarized in Table \ref{tab:sec4}. The table also shows the corresponding updates for $\mathsf{REnKF},$ where $\bar{\upr}^{(j-1)}$, $\barxi^{(j)}$ and $\bareta^{(j)}$ respectively denote the sample means of $\{\upr_n^{(j-1)}\}_{n=1}^N$, $\{\xi_n^{(j)}\}_{n=1}^N,$ and $\{\eta_n^{(j)}\}_{n=1}^N;$   $S^{(j-1)}$ denotes the empirical covariance of $\{\upr_n^{(j-1)}\}_{n=1}^N;$ and $\whatC^{(j)}_{\, u \xi} 
    = (\whatC^{(j)}_{\, \xi u})^\top$ denotes the empirical cross-covariance of $\{\upr_n^{(j-1)}\}_{n=1}^N$ and $\{\xi_n^{(j)}\}_{n=1}^N.$ Finally, following \cite{furrer2007estimation, ghattas2022non}, we refer to 
   \begin{align*}
           \whatO^{(j)}
     &:=  
         \msK(\hatCpr^{(j)}) (\widehat{\Gamma}^{(j)} - \Gamma) \msK^\top(\hatCpr^{(j)})\\
         & \qquad \qquad \qquad  + \bigl(I-\msK(\hatCpr^{(j)}) H\bigr) \whatC^{(j)}_{\, u \eta} \msK^\top(\hatCpr^{(j)}) + 
        \msK(\hatCpr^{(j)}) (\whatC^{(j)}_{\, u \eta})^\top \bigl(I-H^\top\msK^\top(\hatCpr^{(j)})\bigr).
   \end{align*}
    as the \emph{offset}, where $\widehat{\Gamma}^{(j)}$  denotes the empirical covariance of $\{\eta_n^{(j)}\}_{n=1}^N,$ and $\whatC^{(j)}_{\, u \eta }$ denotes the empirical cross-covariance of $\{\upr_n^{(j-1)}\}_{n=1}^N$ and $\{\eta_n^{(j)}\}_{n=1}^N.$

\begin{table} 
\renewcommand{\arraystretch}{1.5}
\centering
\begin{tabular}{c|c|c}
\hline
 & {Kalman Filter} & {$\mathsf{REnKF}$} \\
\hline \hline
{Forecast Mean} & $\mpr^{(j)} = A \mpost^{(j-1)}$ & $\hatmpr^{(j)} = A \bar{\upr}^{(j-1)} + \barxi^{(j)}$ \\
{Forecast Cov.} &  $\Cpr^{(j)} = A \Cpost^{(j-1)} A^\top + \Xi$ & $\hatCpr^{(j)} = A S^{(j-1)} A^\top + \hatxiCovariance^{(j)} + A \whatC^{(j)}_{\, u \xi} + \whatC^{(j)}_{\, \xi u}A^\top$  \\
{Analysis Mean} & $\mpost^{(j)} = \msM(\mpr^{(j)}, \Cpr^{(j)} ; y^{(j)})$ & $\hatmpost^{(j)} = \msM(\hatmpr^{(j)}, \hatCpr^{(j)} ; y^{(j)}) + \msK(\hatCpr^{(j)}) \bareta^{(j)}$  \\
{Analysis Cov.} &  $\Cpost^{(j)} = \msC(\Cpr^{(j)})$ & $\hatCpost^{(j)} = \msC(\hatCpr^{(j)}) + \hatO^{(j)}$ \\
\hline
\end{tabular}
\vspace{0.25cm}
\caption{Kalman filter and $\mathsf{REnKF}$ updates in terms of the operators \eqref{eq:KalmanGainOperator}, \eqref{eq:MeanOperator}, and \eqref{eq:CovarianceOperator}. 
}
\label{tab:sec4}
\end{table}

\begin{remark}[Deterministic Implementations] \label{rem:DeterministicImplementations}
    As noted earlier, our presentation and analysis will focus on the stochastic (perturbed observation) implementation of $\mathsf{EnKF}$ described in Algorithm \ref{algEnKF}, and which is used within $\mathsf{REnKF}$, see Algorithm \ref{algEnKFresample}. We claim that this approach is sufficient to cover both deterministic and stochastic updates. Indeed, \cite{ghattas2022non} shows that deterministic and stochastic updates at time $j$ can be succinctly written as
    \begin{align}\label{eq:Combinedupdate}
    \begin{split}
       \hatmpost^{(j)} &= \msM(\hatmpr^{(j)}, \hatCpr^{(j)} ) + \phi\msK(\hatCpr^{(j)}) \bareta^{(j)},  \\
       \hatCpost^{(j)} &= \msC(\hatCpr^{(j)}) + \phi \widehat{O}^{(j)},
    \end{split}
    \end{align}
    where $\phi=1$ for the stochastic update 
    and $\phi = 0$ for the deterministic update.
    Therefore, relative to the deterministic update, theory for the stochastic update is additionally complicated by the need to consider the term $\msK(\hatCpr^{(j)}) \bareta^{(j)}$ in the mean update and the offset term $\widehat{O}^{(j)}$ in the covariance update. Accordingly, we are able to provide a result for the resampled version of the deterministic (square-root) $\mathsf{EnKF}$ as a by-product of our more general theory, and we refer to Remark~\ref{as} for further discussion.
\end{remark}

\subsubsection{Main Result}\label{ssec:mainresult}
We define the \emph{effective dimension} \cite{tropp2015introduction} of a matrix $Q \in \mcS^d_+$ by 
 \begin{align} \label{eq:effectiveDim}
    r_2(Q) := \frac{\ttrace(Q)}{|Q|},   
 \end{align}
 where $\ttrace(Q)$ and $|Q|$ denote the trace and operator norm of $Q$. The effective dimension quantifies the number of directions where $Q$ has significant spectral content and may be significantly smaller than the ambient dimension $d$ when the eigenvalues of $Q$ decay quickly. As such, it is a more refined measure of complexity in high-dimensional problems with underlying low-dimensional structure. The monographs \cite{tropp2015introduction,vershynin2018high} refer to $r_2(Q)$ as the intrinsic dimension, while \cite{koltchinskii2017concentration} uses the term effective rank. This terminology is motivated by the observation that $1 \le r_2(Q) \le \text{rank} (Q) \le d $ and that $r_2(Q)$ is insensitive to changes in the scale of $Q$, see \cite{tropp2015introduction}. We now state our main result, Theorem~\ref{thm:MultiStepPOEnKFBounds}, which provides non-asymptotic bounds on the deviation of $\mathsf{REnKF}$ from the Kalman filter for any time $j$.
 
 \begin{theorem}\label{thm:MultiStepPOEnKFBounds}
 Consider $\mathsf{REnKF},$ Algorithm \ref{algEnKFresample}, with linear dynamics $\Psi(\cdot) = A\cdot$. Suppose that $N \ge r_2(\Cpost^{(0)}) \lor r_2(\Gamma) \lor r_2(\xiCovariance)$. For any $j= 1,2,\dots$, and $q \ge 1$
    \begin{align}
        \||\hatmpost^{(j)} - \mpost^{(j)}| \|_q
        &\le    
        c_1
        \inparen{
        \sqrt{\frac{r_2(\Cpost^{(0)})}{N}}
        \lor 
        \sqrt{\frac{r_2(\xiCovariance)}{N}}
        \lor \sqrt{\frac{r_2(\Gamma)}{N}}
        },  \label{eq:meanboundmainth}   \\
        \normn{ | \hatCpost^{(j)} - \Cpost^{(j)}| }_{q}
        &\le 
        c_2
        \inparen{
        \sqrt{\frac{r_2(\Cpost^{(0)})}{N}}
        \lor
        \sqrt{\frac{r_2(\xiCovariance)}{N}}
        \lor 
        \sqrt{\frac{r_2(\Gamma)}{N}}  
        }, \label{eq:covarianceboundmainth}
    \end{align}
    where $\mpost^{(j)}$ and $\Cpost^{(j)}$ are the mean and covariance of the filtering distributions, and $c_1, c_2$ are potentially different universal constants depending on
    \begin{align*}
        |\Cpost^{(0)}|, |A|, |H|, |\Gamma^{-1}|, |\Gamma|, |\xiCovariance|, q, j,
    \end{align*}
    and $c_1$ additionally depends on  $\{|y^{(\ell)}-H \mpr^{(\ell)}| \}_{\ell \le j}$.
    \end{theorem}
   With the exception of \cite[Theorem 3.4]{majda2018performance}, which relies on covariance inflation and an additional projection step, Theorem~\ref{thm:MultiStepPOEnKFBounds} seems to be the first result in the literature that provides non-asymptotic guarantees on the performance of a stochastic  $\mathsf{EnKF}$ over multiple assimilation cycles. 
    We note that the assumption $N \ge r_2(\Cpost^{(0)}) \lor r_2(\Gamma) \lor r_2(\xiCovariance)$ is merely for convenience and can be removed at the expense of a more cumbersome statement of the result. Importantly, the bounds \eqref{eq:meanboundmainth} and \eqref{eq:covarianceboundmainth} are non-asymptotic, in that they hold for a fixed ensemble size $N$. Further, the bounds are dimension-free as they do not exhibit any dependence on the state-space dimension $d$, implying that the ensemble need not scale with $d$ in order for the algorithm to perform well, as has been observed empirically in the literature and confirmed in our numerical results in Section~\ref{sec:numerics}. 
    Finally, similar to previous accuracy analyses for square-root ensemble Kalman filters \cite{mandel2011convergence,ghattas2022non}, variational data assimilation algorithms \cite{sanz2015long,law2016filter}, and particle filters \cite[Chapters 11 and 12]{sanzstuarttaeb2023}, our proof relies on induction over the discrete time index $j$ and does not account for potential dissipation of errors due to filter ergodicity. As a result, the constants $c_1$ and $c_2$ grow with $j$ and our bounds \eqref{eq:meanboundmainth} and \eqref{eq:covarianceboundmainth} do not hold uniformly in time without, for instance, stability requirements on $A.$

    \begin{remark}[Resampled Square-Root Filter]\label{as}
    While the result in Theorem~\ref{thm:MultiStepPOEnKFBounds} is specific to the stochastic $\mathsf{REnKF}$ in Algorithm \ref{algEnKFresample}, using the observation made in Remark~\ref{rem:DeterministicImplementations} it is possible to show that for a deterministic variant, namely the square-root $\mathsf{REnKF}$, and under the same assumptions on the ensemble size made in Theorem~\ref{thm:MultiStepPOEnKFBounds}, we have that 
    \begin{align} \label{eq:SRBounds}
    \begin{split}
        \||\hatmpost^{(j)} - \mpost^{(j)}| \|_q
        &\le    
        c_1
        \inparen{
        \sqrt{\frac{r_2(\Cpost^{(0)})}{N}}
        \lor 
        \sqrt{\frac{r_2(\xiCovariance)}{N}}
        }, \\ 
        \normn{ | \hatCpost^{(j)} - \Cpost^{(j)}| }_{q}
        &\le 
        c_2
        \inparen{
        \sqrt{\frac{r_2(\Cpost^{(0)})}{N}}
        \lor
        \sqrt{\frac{r_2(\xiCovariance)}{N}}
        }, 
        \end{split}
    \end{align}
    where $c_1, c_2$ are potentially different universal constants depending on $|\Cpost^{(0)}|$, $|A|$, $|H|$, $|\Gamma^{-1}|$, $|\xiCovariance|$, $q$, $j$, and $c_1$ additionally depends on  $\{|y^{(\ell)}-H \mpr^{(\ell)}| \}_{\ell \le j}$. In contrast to \eqref{eq:meanboundmainth} and \eqref{eq:covarianceboundmainth}, the bounds in \eqref{eq:SRBounds} do not depend on the effective dimension of the noise covariance, $r_2(\Gamma)$, nor do the associated constants depend on $|\Gamma|$. The statistical price to pay for utilizing stochastic rather than deterministic updates is captured by these terms. We further note that \cite[Corollary A.12]{ghattas2022non}, gives a non-asymptotic and multi-step analysis of a simplified version of the square-root filter (without resampling) with deterministic dynamics (that is, $\xiCovariance=O_{d\times d}$). In such a setting, \cite[Corollary A.12]{ghattas2022non}  implies the following bounds
    \begin{align*}
        \||\hatmpost^{(j)} - \mpost^{(j)}| \|_q
        \le    
        c_3
        \sqrt{\frac{r_2(\Cpost^{(0)})}{N}}
        ,  \qquad 
        \normn{ | \hatCpost^{(j)} - \Cpost^{(j)}| }_{q}
        \le 
        c_4
        \sqrt{\frac{r_2(\Cpost^{(0)})}{N}}, 
    \end{align*}
    where $c_3, c_4$ are potentially different universal constants depending on $|\Cpost^{(0)}|$, $|A|$, $|H|$, $|\Gamma^{-1}|$, $q$, $j$, and $c_3$ additionally depends on  $\{|y^{(\ell)}-H \mpr^{(\ell)}| \}_{\ell \le j}$. 
    Theorem \ref{thm:MultiStepPOEnKFBounds} should further be compared to \cite[Theorem 6.1]{kwiatkowski2015convergence}, which is also limited to the case $\xiCovariance=O_{d\times d}$ and shows that $ \||\hatmpost^{(j)} - \mpost^{(j)}| \|_q
        \le    
        c_3'N^{-1/2}$ and $ \normn{ | \hatCpost^{(j)} - \Cpost^{(j)}| }_{q}
        \le 
        c_4'N^{-1/2},$
    where $c_3', c_4'$ are 
    universal constants with the same dependencies as $c_3$ and $c_4.$
    Importantly, the bounds in \cite[Theorem 6.1]{kwiatkowski2015convergence} do not capture the dependence of the algorithm on the  prior covariance
    and also cannot be easily extended to handle stochastic dynamics $\xiCovariance \succ 0$ as accomplished in Theorem \ref{thm:MultiStepPOEnKFBounds}. 
    \end{remark}

\section{Numerical Results}\label{sec:numerics}
In this section, we investigate the empirical performance of $\mathsf{REnKF}$ (Algorithm~\ref{algEnKFresample}) and provide detailed comparisons to the stochastic $\mathsf{EnKF}$ (Algorithm~\ref{algEnKF}). In Subsection~\ref{sec:linearexample}, we consider a linear dynamics map, $\Psi(\cdot) = A \cdot$, with the primary goal of demonstrating the bounds of Theorem~\ref{thm:MultiStepPOEnKFBounds} in simulated settings. In Subsection~\ref{sec:l96example}, we study a nonlinear setting where $\Psi$ represents the $\Delta t$-flow of the Lorenz 96 system, and $\Delta t$ is the (constant) time-span between observations. The aim of this subsection is to show that $\mathsf{REnKF}$ achieves comparable performance to $\mathsf{EnKF}$ even in challenging nonlinear regimes, further motivating the study of resampling in the context of ensemble algorithms. In both Subsections~\ref{sec:linearexample} and \ref{sec:l96example}, we examine the performance of $\mathsf{REnKF}$ and $\mathsf{EnKF}$ under varying noise levels, ensemble sizes, and state dimensions. Additionally, in Subsection~\ref{sec:l96example}, we consider cases in which we have access to either fully observed or partially observed dynamics. These scenarios offer a comprehensive perspective on the adaptability of $\mathsf{REnKF}$ to varying observational conditions, thereby highlighting its potential for wide applicability in real-world situations where data are often limited or incomplete.  For all experiments, we generate a ground-truth state process $\{u^{(j)}\}_{j=0}^J$ for a time-window of length $J=200$ using the initialization \eqref{eq:HMMinitialization} and dynamics model \eqref{eq:HMMdynamics}. For each set of system parameters we examine, a unique set of observations $\{y^{(j)} \}_{j=1}^J$ is generated from the ground-truth state process utilizing the observation model \eqref{eq:HMMobservations}. Python code to reproduce all numerical experiments is publicly available at \href{https://github.com/Jiajun-Bao/EnKF-with-Resampling}{https://github.com/Jiajun-Bao/EnKF-with-Resampling}.

\subsection{Linear Dynamics} \label{sec:linearexample}
In this subsection, we numerically investigate the performance of $\mathsf{REnKF}$ for the linear-Gaussian hidden Markov model \eqref{eq:linearinit}-\eqref{eq:linearobservations} analyzed in Subsection~\ref{sec:theory}. We will consider a variety of choices for the initial distribution, the dynamics noise covariance, and the observation noise covariance. Throughout, we take identity dynamics $A = I_d$ and full observations $H = I_d.$ To compare the performance of  $\mathsf{EnKF}$ and $\mathsf{REnKF},$ we will consider the following metrics:
\begin{alignat}{3}
 &\text{(Mean Error)}  \qquad   && \hspace{0.5cm} \mathsf{E}_{\text {\tiny Linear} }  &&= 
\frac{1}{J} \sum _{j=1}^J |\hatmpost^{(j)} - \mpost^{(j)}|_2, \label{eq:meanerror}\\
 &\text{(CI Width)}    && \hspace{1cm} \mathsf{W} &&= \frac{1}{J} \sum_{j=1}^J \frac{1}{d} \sum_{i=1}^d 2 \times 1.96 \sqrt{ \widehat{\Sigma}_{ii}^{(j)}} \, , \label{eq:CIwidth}\\
 &\text{(CI Coverage)}     &&  \hspace{1.1cm} \mathsf{V} &&= \frac{1}{J} \sum_{j = 1}^J \frac{1}{d} \sum_{i=1}^d \indicator \bigl\{u^{(j)}(i) \in (\widehat{\mu}^{(j)}(i) \pm 1.96 \sqrt{ \widehat{\Sigma}_{ii}^{(j)}} ) \bigr\}.  \label{eq:CIcoverage}
\end{alignat}
\begin{figure} 
    \centering
    \includegraphics[width=\textwidth]{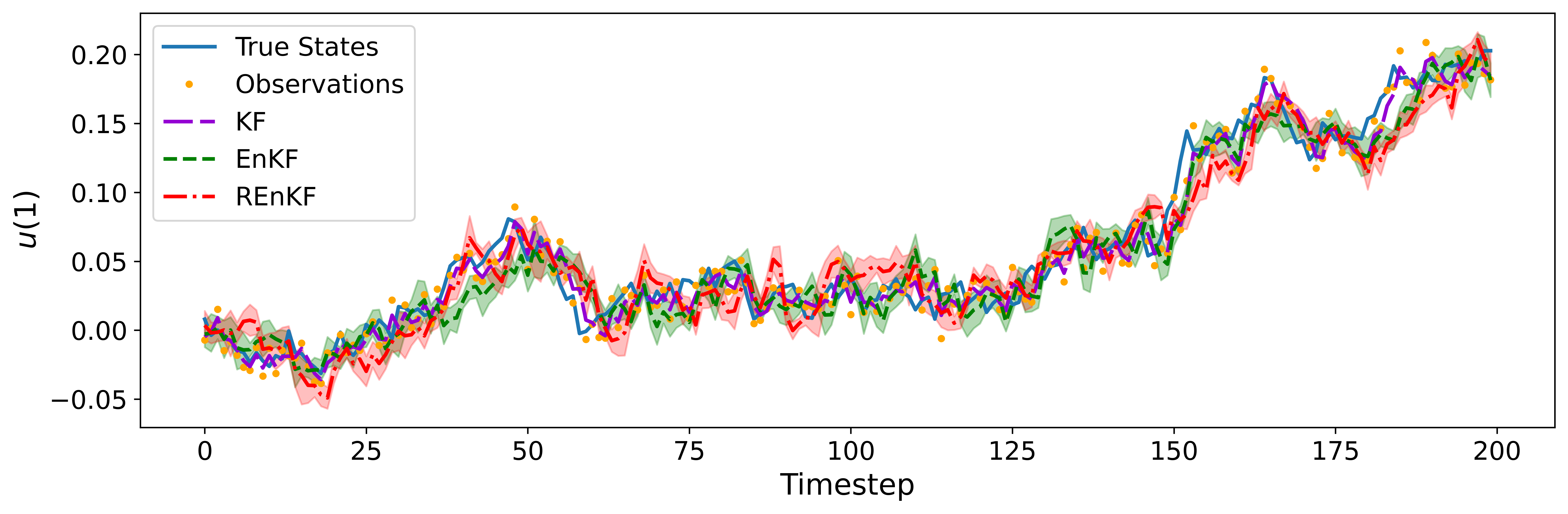}
    \vspace{-.65cm}
    \caption{State estimation and uncertainty quantification for coordinate $u(1)$ in the linear setting with ensemble size $N = 10$ and small noise $\alpha = 10^{-4}.$ Note that the Kaman Filter (KF) is optimal in the linear setting.
    }
    \label{fig:linear_one_simulation}
\end{figure}

The mean error \eqref{eq:meanerror} quantifies the approximation of the $\mathsf{EnKF}/\mathsf{REnKF}$ analysis mean to the mean $\mu^{(j)}$ of the Kalman filter. Our theory for $\mathsf{REnKF}$ provides non-asymptotic bounds for this error, and our numerical results will show that this error is similar to that of $\mathsf{EnKF}$ in a variety of settings. The confidence interval (CI) width and coverage in \eqref{eq:CIwidth}-\eqref{eq:CIcoverage} assess the ability of the filter to provide reliable uncertainty quantification: a short interval with high coverage would be preferable, but an overconfident short width interval with low coverage can lead to a misleading and potentially dangerous assessment of uncertainty. 
We illustrate these three metrics in Figure \ref{fig:linear_one_simulation}, which corresponds to a setup outlined in Table \ref{tab:linear}. This setup will be further explored in Subsection \ref{sec:linear_sub1}.  As depicted in the plot, the indicator in \eqref{eq:CIcoverage} corresponds to whether the solid blue line (representing the true states) fall within the shaded confidence intervals.
We point out that the ability of ensemble Kalman methods to provide reliable uncertainty quantification, especially in nonlinear settings, has often been questioned \cite{ernst2015analysis,law2012evaluating}. Our results will show that the CIs obtained with $\mathsf{REnKF}$ have similar width and coverage as those obtained by $\mathsf{EnKF},$ but that coverage for both algorithms is not reliable when the ensemble size is small (Subsections \ref{sec:linearexample} and \ref{sec:l96example}) or the dynamics are highly nonlinear (Subsection \ref{sec:l96example}). 

Since the outputs $\{\hatmpost^{(j)}, \hatCpost^{(j)}\}_{j=1}^J$ of $\mathsf{EnKF}$ and $\mathsf{REnKF}$ are random, for each experiment we run both algorithms $M$ times and we report the average value of the metrics \eqref{eq:meanerror}, \eqref{eq:CIwidth}, and \eqref{eq:CIcoverage} as well as the value of $M$. More details can be found in Appendix \ref{sec:metrics}.

\subsubsection{Effects of Noise Level and Ensemble Size}
\label{sec:linear_sub1}
We perform two distinct analyses to assess the impact of different variables on the performance of $\mathsf{EnKF}$ and $\mathsf{REnKF}$. The first, which we term the \textit{noise-level} analysis, investigates the relationship between mean error, $\mathsf{E}_{\text {\tiny Linear}}$, and the noise level, $\alpha$. The second analysis, referred to as the \textit{ensemble-size analysis}, explores how the mean error varies with the ensemble size, $N$. Both analyses are carried out using a fixed state dimension $d=20$. 

In the noise-level analysis, $\alpha$ is varied over a grid of $15$ evenly spaced values between $10^{-16}$ and $1$,
allowing us to investigate a range of scenarios beginning with those with virtually no noise to those with substantial noise. In order to isolate the influence of $\alpha$, we maintain the initial distribution with a fixed zero mean and covariance $\Sigma^{(0)} = 10^{-8} \times I_{20}$, as well as a fixed ensemble size of $N=20$. In the ensemble-size analysis, $N$ is varied between 10 and 100, in increments of 10. To isolate the effects of $N$, we fix $\alpha=10^{-1}$ and maintain the initial distribution to have a fixed zero mean and covariance $\Sigma^{(0)} = 1.1\alpha \times I_{20}$. The covariance is adjusted to represent a higher initial uncertainty level compared to the noise-level analysis. The factor $1.1$ was introduced to ensure that the initial states possess a slightly different level of uncertainty relative to the noise in the dynamics and observations. Both analyses are averaged over $M = 10$ runs of the algorithms. The results of both analyses are depicted in Figure \ref{fig:linear_plot_1}. 
\begin{figure} 
    \centering
    \includegraphics[width=\textwidth]{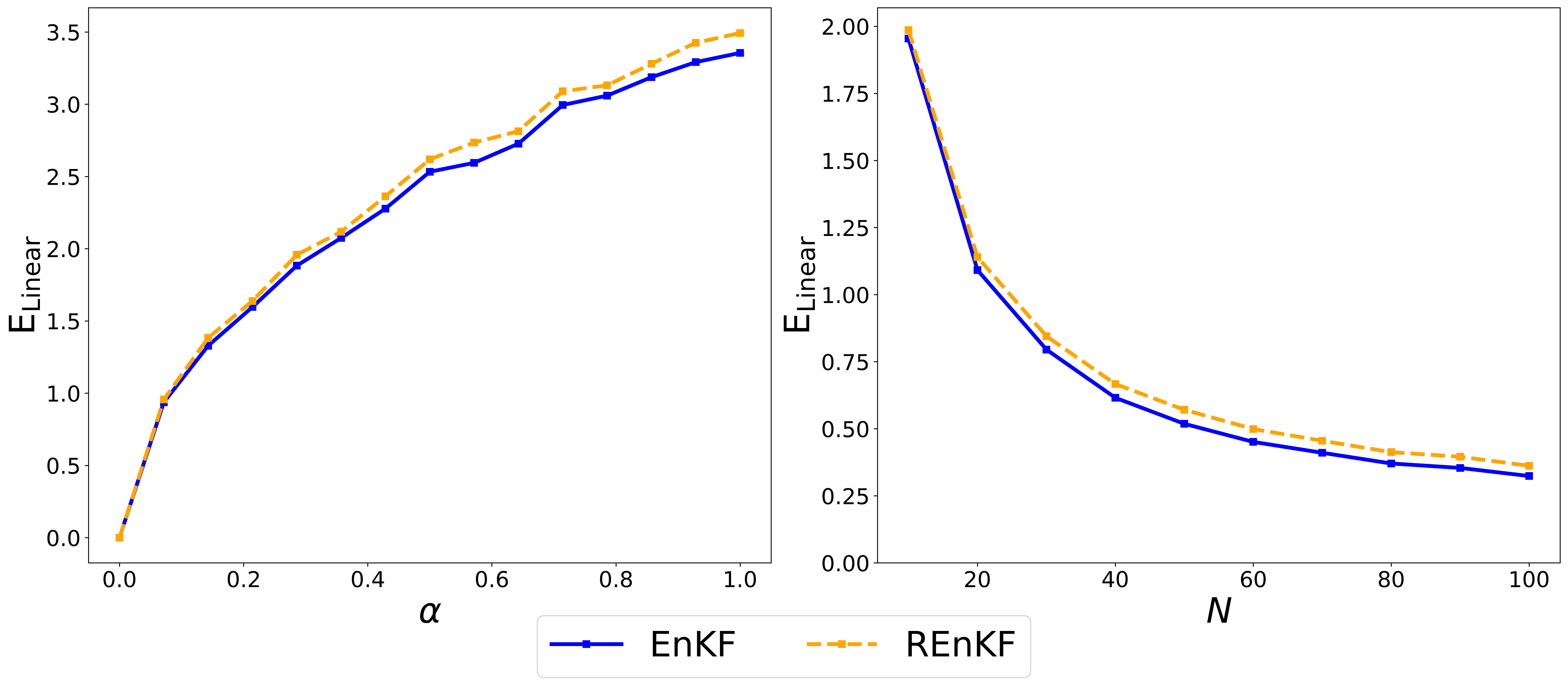}
    \vspace{-.65cm}
    \caption{Effects of $\alpha$ and $N$ in the linear setting with $d=20$.}
    \label{fig:linear_plot_1}
\end{figure}
In addition to $\mathsf{E}_{\text {\tiny Linear}}$, in Table \ref{tab:linear} we consider the effect of varying $\alpha$ and $N$ on CI  widths, $\mathsf{W}$, and CI  coverage, $\mathsf{V}$. Here, we categorize the levels of noise as being either small, moderate, or large, which correspond to $\alpha$ values of $10^{-4}$, $10^{-2}$, or $10^{-1}$ respectively, as described under Case A in Table \ref{tab:noise_covariance}. Further, we repeat the experiments with ensembles of size $N=10$ and $N=40$. For the experimental settings summarized in Table~\ref{tab:linear}, the state dimension and initial distribution are taken as in the ensemble-size analysis described earlier. These metrics are calculated based on averages over $M=100$ runs of the algorithms. 

The results in Figure~\ref{fig:linear_plot_1} and in Table~\ref{tab:linear} confirm that across a wide variety of linear experimental settings, $\mathsf{REnKF}$ exhibits similar performance to $\mathsf{EnKF}$ as measured by the mean error, CI  width, and CI coverage.

\begin{table} 
\centering
\resizebox{0.8\textwidth}{!}{%
\begin{tabular}{c|c|c|c|c}
\hline 
 \multirow{2}{*}{Ensemble} & \multirow{2}{*}{Metric} & Small Noise & Moderate Noise & Large Noise \\
  & & $\alpha = 10^{-4} $ & $\alpha = 10^{-2}$ & $\alpha = 10^{-1}$   \\ 
\hline\hline
\multirow{6}{*}{$N=10$} &
$\mathsf{EnKF}$ Mean Error & 0.0608 & 0.6133 & 1.9931 \\
& $\mathsf{REnKF}$ Mean Error & 0.0616 & 0.6199 & 2.0310 \\
& $\mathsf{EnKF}$ CI Width & 0.0194 & 0.1940 & 0.6134 \\
& $\mathsf{REnKF}$ CI Width & 0.0188 & 0.1875 & 0.5930 \\
& $\mathsf{EnKF}$ CI Coverage (\%) & 39.57 & 38.90 & 38.35 \\
& $\mathsf{REnKF}$ CI Coverage (\%) & 37.83 & 37.14 & 36.58 \\
\hline\hline
\multirow{6}{*}{$N=40$} &
$\mathsf{EnKF}$ Mean Error & 0.0193 & 0.1930 & 0.6243 \\
& $\mathsf{REnKF}$ Mean Error & 0.0209 & 0.2091 & 0.6739 \\
& $\mathsf{EnKF}$ CI  Width & 0.0278 & 0.2780 & 0.8790 \\
& $\mathsf{REnKF}$ CI  Width & 0.0274 & 0.2739 & 0.8663 \\
& $\mathsf{EnKF}$ CI  Coverage (\%) & 69.94 & 69.26 & 68.90 \\
& $\mathsf{REnKF}$ CI  Coverage (\%) & 68.65 & 67.76 & 67.43 \\
\hline
\end{tabular}%
}
\vspace{0.25cm}
\caption{Performance metrics in the linear setting with $d=20.$}
\label{tab:linear}
\end{table}

\subsubsection{Effects of State Dimension and Spectrum Decay} 
\label{sec:linear_sub2}
We now study the sensitivity of $\mathsf{EnKF}$ and $\mathsf{REnKF}$ to changes in the state dimension, $d$.
Recall that our main result, Theorem~\ref{thm:MultiStepPOEnKFBounds}, implies that $\mathsf{REnKF}$ performs well whenever the ensemble size scales with the largest of the effective dimensions of the noise covariances: $\Sigma^{(0)}$, $\Gamma,$ and $\xiCovariance$. This motivates our study of covariance matrices with structure summarized in Case A and Case B of Table~\ref{tab:noise_covariance}. 
\begin{table} [ht]
\renewcommand{\arraystretch}{1.5}
\centering
\resizebox{0.825\textwidth}{!}{%
\begin{tabular}{c|c|c|c}
\hline
 Noise & {Case A} & {Case B} $(i=1,\dots, d)$ &{Case C} \\
\hline\hline
{Dynamics ($\Xi$)} & $\Xi^{A} = \alpha \times I_d$ & $\Xi^{B}_{ii} = \alpha \times  i^{-\beta}$ & $\Xi^{C} = \alpha \times I_d$\\
{Observation ($\Gamma$)} & $\Gamma^{A} = \alpha \times I_d$ & $\Gamma^{B}_{ii} = \alpha \times  i^{-\beta}$& $\Gamma^{C} = \alpha \times I_{\frac{2d}{3}}$ \\
{Prior ($\Sigma^{(0)}$)} & $(\Sigma^{(0)})^{A} = 1.1 \times \Xi^{A} $ & $(\Sigma^{(0)})^{B}_{ii} = 1.1 \times \Xi^B_{ii} $ &
 $(\Sigma^{(0)})^{C} = 1.1 \times \Xi^{C} $\\
\hline
\end{tabular}
}
\vspace{0.25cm}
\caption{Covariance matrix settings explored numerically in Subsections~\ref{sec:linear_sub2} and \ref{sec:l96example}.}
\label{tab:noise_covariance}
\end{table}
In Case A, the effective dimension of the covariance matrix is proportional to the state dimension, $d$, and so the theory suggests that $\mathsf{REnKF}$ will do well only if the ensemble size also scales with $d$. In Case $B$, we consider covariance matrices that are diagonal, with $i$-th diagonal element proportional to $i^{-\beta}$ where $\beta>0$ is a rate parameter controlling the speed of decay. 
\begin{table} 
\centering
\resizebox{0.8\textwidth}{!}{%
\begin{tabular}{c|c|c|c|c|c|c|c|c }
\hline
State dimension ($d$) & $2$ & $4$ & $8$ & $16$ & $32$ & $64$ & $128$ & $256$\\ 
\hline\hline
$\beta$ = 0.1 & 1.93 & 3.70 & 7.02 & 13.25 & 24.89 & 46.64 & 87.25 & 163.05\\
$\beta$ = 1.0 & 1.50 & 2.08 & 2.72 & 3.38 & 4.06 & 4.74 & 5.43 & 6.12\\
$\beta$ = 1.5 & 1.35 & 1.67 & 1.93 & 2.12 & 2.26 & 2.36 &  2.44 & 2.49\\
\hline
\end{tabular}
}
\vspace{0.25cm}
\caption{Effective dimension of initialization and noise covariances used in Figure \ref{fig:linear_plot_2}.}
\label{tab:linear_noise_cov_effective_dim}
\end{table}
Table~\ref{tab:linear_noise_cov_effective_dim} demonstrates that two matrices of this form that are equal in dimension may differ drastically in their effective dimension for different choices of $\beta$. Here, then, the theory suggests that $\mathsf{REnKF}$ will do well so long as the ensemble size scales with the effective dimension, which may be much smaller than $d$. To test our theory, we run $\mathsf{REnKF}$ under both cases A and B in Table~\ref{tab:noise_covariance} where $d$ is varied over the set $\{2^1,2^2,\dots, 2^8\}$ and where the ensemble size is fixed at $N=10$ throughout. For both cases we fix $\alpha = 10^{-4}$ and for case B we consider $\beta \in \{0.1, 1, 1.5\}$. Figure \ref{fig:linear_plot_2} presents the results of averaging $\mathsf{E}_{\text {\tiny Linear}}$ over $M = 10$ runs of the algorithm in each of the experimental set-ups. We see that for all choices of $\beta$, $\mathsf{EnKF}$ and $\mathsf{REnKF}$ exhibit near-identical performance. For Case A, the performance deteriorates as $d$ increases and this behavior is identical across all three displays. For Case $B$, when $\beta=0.1$ (first display) so that the effective dimension increases significantly with dimension as described in the first row of Table~\ref{tab:linear_noise_cov_effective_dim}, the performance deteriorates significantly as $d$ increases. As $\beta$ is increased to $1$ in the second display, so that the effective dimension grows slowly with $d$, performance deteriorates at a much slower rate. This is further pronounced in the final display with $\beta=1.5$.
These numerical results demonstrate the key role played by the effective dimension in determining the performance of $\mathsf{EnKF}$ and $\mathsf{REnKF}$, and are in agreement with Theorem~\ref{thm:MultiStepPOEnKFBounds} for $\mathsf{REnKF}$.

\begin{figure}[htbp]
    \centering
    \includegraphics[width=\textwidth]{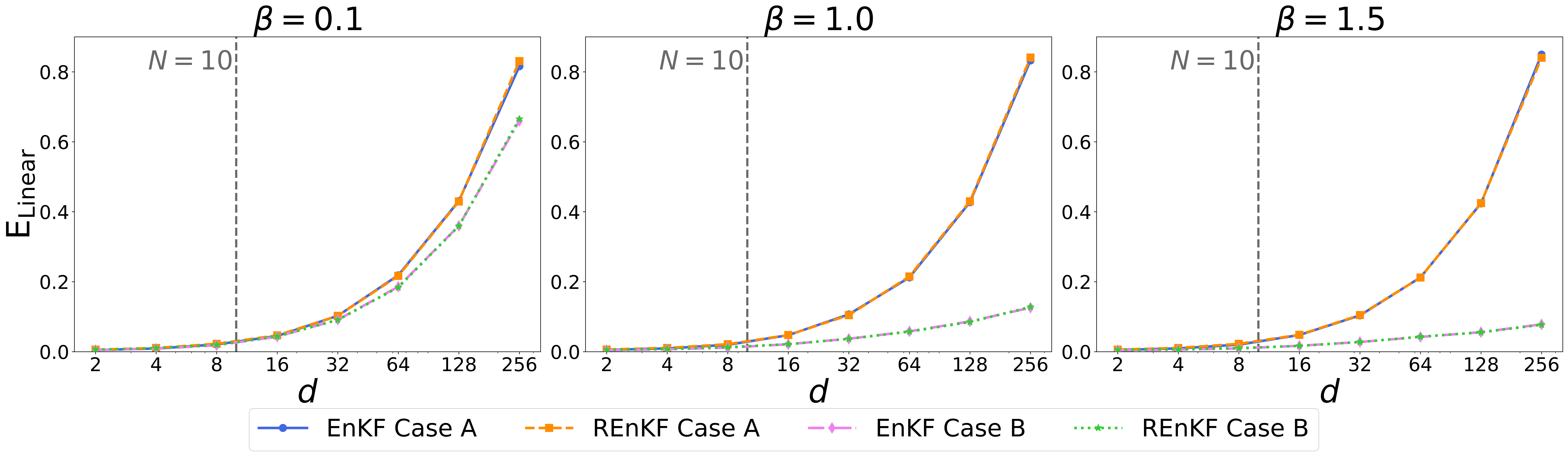}
    \vspace{-.65cm}
    \caption{Effect of spectrum decay in the linear setting.}
    \label{fig:linear_plot_2}
\end{figure}

\subsection{Lorenz 96 Dynamics} \label{sec:l96example}
In this subsection, we extend our numerical investigation of $\mathsf{REnKF}$ to the nonlinear setting by taking $\Psi$ in \eqref{eq:HMMdynamics} to be the $\Delta t$-flow of the Lorenz 96 equations. Here $\Delta t$ represents the time-span between observations, which is assumed to be constant. Assuming the following cyclic boundary conditions $u(-1) = u(d-1)$, $u(0)=u(d)$, and $u(d+1) = u(1)$ with $d\geq4$, the system is governed by:
    \begin{align}\label{eq:lorenz96}
        \frac{du(i)}{dt} = \Bigl(u(i+1) - u(i-2)\Bigr)u(i-1) - u(i) + F,
        \quad 
        i = 1, \ldots, d.
    \end{align}
In our experiments, we set $\Delta t = 0.01,$ $F=8$, and the state dimension $d$ is subject to variation. The choice $F=8$ leads to strongly chaotic turbulence, which hinders predictability in the absence of observations \cite{majda2012filtering}. For the observation process \eqref{eq:HMMobservations}, we consider both full observations in which $H=I_d$, and partial observations in which only two out of every three state components are observed. The latter setting results in a modified $H \in \R^{{\frac{2d}{3}} \times d}$ which corresponds to $I_{d}$ with every third row removed. This observation set-up is motivated by \cite{sanz2015long,law2015data}, which prove that observing two-out-of-three coordinates of the Lorenz 96 system suffices in order to tame the unpredictability of the system and achieve long-time filter accuracy in a small noise regime. As in Subsection~\ref{sec:linearexample}, we examine various choices of initial distribution, dynamics noise covariance, and observation noise covariance. To compare $\mathsf{EnKF}$ and $\mathsf{REnKF},$ we make use of the same CI  width \eqref{eq:CIwidth} and CI  coverage \eqref{eq:CIcoverage} metrics as in Subsection \ref{sec:linearexample}. However, since in the nonlinear setting the mean of the filtering distribution is not available in closed form, we replace the metric $\mathsf{E}_{\text {\tiny Linear}}$ with 
\begin{align}
\label{eq:l96meanerror}
\mathsf{E}_{\text {\tiny L96}} = \frac{1}{J} \sum_{j=1}^J |\hatmpost^{(j)} - u^{(j)}|_2,
\end{align}
which quantifies the accuracy of the filter as an estimator of the ground-truth state process $\{u^{(j)}\}_{j=1}^J$. As before, the metrics we report are averaged over $M$ runs of the algorithms.

\begin{table}[ht]
\resizebox{\textwidth}{!}{%
\begin{tabular}{c|c|c|c|c|c|c|c}
\hline
\multicolumn{2}{c|}{} & \multicolumn{3}{c|}{Full Observation} & \multicolumn{3}{c}{Partial Observation} \\
\hline
 \multirow{2}{*}{Ensemble} &\multirow{2}{*}{Metric} & Small Noise  & Moderate Noise & Large Noise & Small Noise & Moderate Noise & Large Noise \\
 & & $\alpha = 10^{-4} $ & $\alpha = 10^{-2}$ & $\alpha = 10^{-1}$ & $\alpha = 10^{-4} $ & $\alpha = 10^{-2}$ & $\alpha = 10^{-1}$  \\ 
\hline\hline
\multirow{6}{*}{$N=21$} &
$\mathsf{EnKF}$ Mean Error & 0.1011 & 0.9573 & 3.0231 & 0.4064 & 3.3882 & 10.5921 \\
& $\mathsf{REnKF}$ Mean Error & 0.1016 & 0.9616 & 3.0335 & 0.4071 & 3.3565 & 10.6379 \\
& $\mathsf{EnKF}$ CI  Width & 0.0208 & 0.2083 & 0.6586 & 0.0266 & 0.2660 & 0.8412 \\
& $\mathsf{REnKF}$ CI  Width & 0.0205 & 0.2047 & 0.6475 & 0.0258 & 0.2584 & 0.8167 \\
& $\mathsf{EnKF}$ CI  Coverage (\%)  & 50.24 & 51.55 & 51.61 & 39.62 & 43.25 & 43.26 \\
& $\mathsf{REnKF}$ CI  Coverage (\%) & 49.07 & 50.34 & 50.44 & 38.25 & 42.04 & 41.87 \\
\hline\hline
\multirow{6}{*}{$N=84$} &
$\mathsf{EnKF}$ Mean Error & 0.0582 & 0.5682 & 1.7971 & 0.2919 & 2.4181 & 7.6282 \\
& $\mathsf{REnKF}$ Mean Error & 0.0590 & 0.5760 & 1.8218 & 0.2977 & 2.5004 & 7.9011 \\
& $\mathsf{EnKF}$ CI  Width & 0.0281 & 0.2813 & 0.8895 & 0.0438 & 0.4383 & 1.3861 \\
& $\mathsf{REnKF}$ CI  Width & 0.0279 & 0.2785 & 0.8806 & 0.0412 & 0.4120 & 1.3033 \\
& $\mathsf{EnKF}$ CI  Coverage (\%) & 87.96 & 88.61 & 88.61 & 71.47 & 75.31 & 75.30 \\
& $\mathsf{REnKF}$ CI  Coverage (\%) & 86.80 & 87.52 & 87.52 & 69.25 & 72.54 &  72.61 \\
\hline
\end{tabular}%
}
\vspace{0.25cm}
\caption{Performance metrics for the Lorenz 96 model with $d=42.$}
\label{tab:nonlinear}
\end{table}
In Table \ref{tab:nonlinear}, we compare the performance of $\mathsf{REnKF}$ and $\mathsf{EnKF}.$ In the case of full observations, the covariance configuration is outlined in Case A of Table \ref{tab:noise_covariance}, and in the case of partial observations it is outlined in Case C of Table \ref{tab:noise_covariance}. 
We repeat the experiments with ensembles of size $N=21$ and $N=84$, and the metrics are computed over $M = 100$ runs of the algorithms. 
In Figure \ref{fig:nonlinear_one_simulation}, we present a single representative simulation of the first component $u(1)$ ---which is observed--- and the third component $u(3)$ ---which is unobserved--- corresponding to a particular choice of parameters in Table \ref{tab:nonlinear}. Additional experiments in the accompanying Github repository show that, as the noise level $\alpha$ increases, state estimation remains effective for observed variables but deteriorates for unobserved ones. This behavior explains the larger error for moderate and large noise levels in the partial observation set-up in  Table \ref{tab:nonlinear}.

\begin{figure}[htbp]
    \centering
    \includegraphics[width=\textwidth]{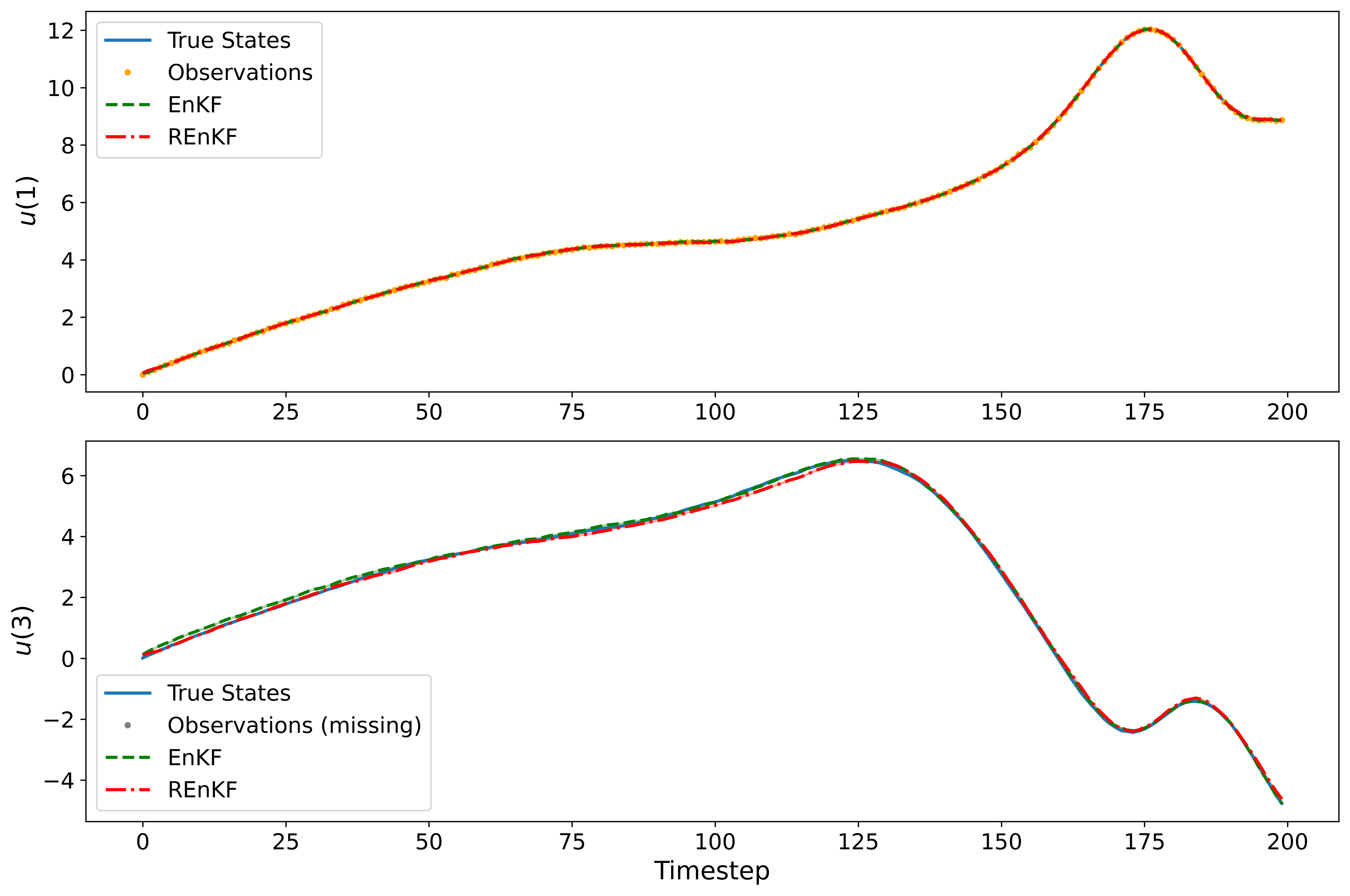}
    \vspace{-0.7cm}
    \caption{State estimation of coordinates $u(1)$ (observed) and $u(3)$ (unobserved) in a partially observed Lorenz 96 system with ensemble size $N = 21$ and small noise $\alpha = 10^{-4}.$ $\mathsf{REnKF}$ accurately recovers observed and unobserved coordinates of the state.}
    \label{fig:nonlinear_one_simulation}
\end{figure}

In Figure~\ref{fig:nonlinear_plot}, we further analyze the effects of varying $\alpha$ (column 1), $N$ (column 2), and $d$ (column 3) on $\mathsf{E}{\text {\tiny L96}}$ in both the full observation (row 1) and partial observation (row 2) settings. More precisely, in the first column of Figure~\ref{fig:nonlinear_plot}, $\alpha$ is varied over a grid of $15$ evenly spaced values between $10^{-16}$ and $1$ while holding fixed $N=20$ and $d=42$. In both full and partial observation settings, we take the initial distribution to have zero mean and covariance $\Sigma^{(0)} = 10^{-8} \times I_{42}$. 
\begin{figure} 
    \centering
    \includegraphics[width=\textwidth]{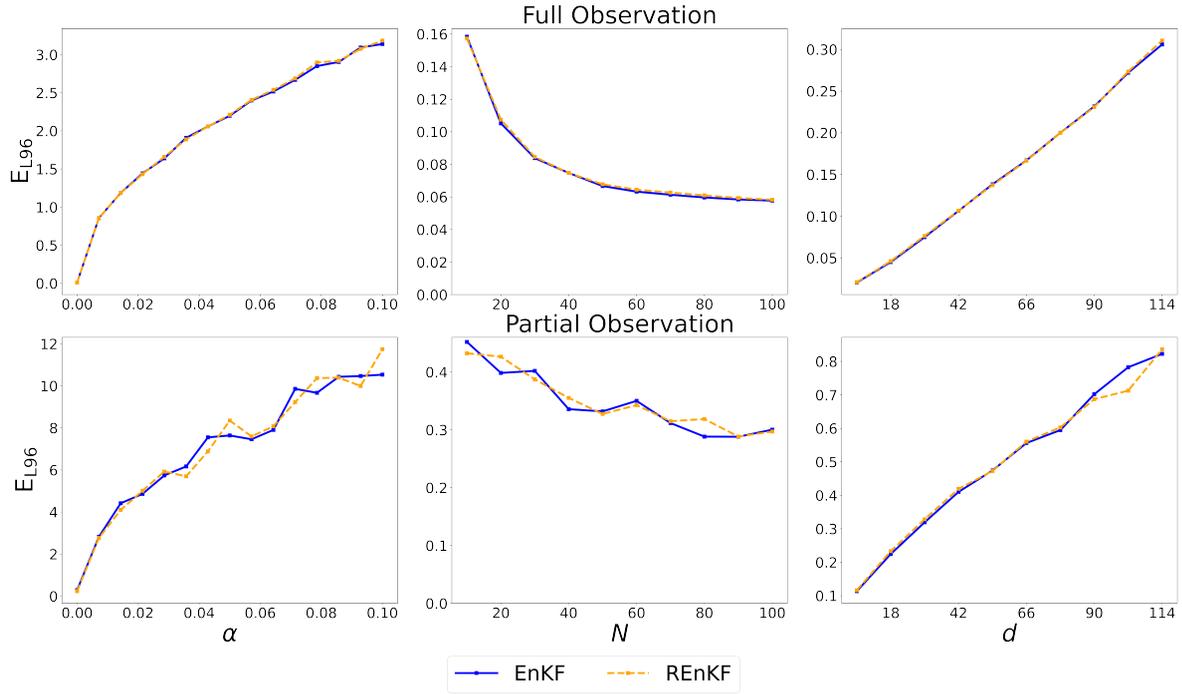}
    \vspace{-.65cm}
    \caption{Effects of $\alpha$, $N,$ and $d$ in the Lorenz 96 example.}
    \label{fig:nonlinear_plot}
\end{figure}
In the second column of Figure~\ref{fig:nonlinear_plot}, the ensemble size $N$ ranges from 10 to 100, increasing in steps of 10, while fixing $\alpha=10^{-4}$ and $d=42$. In both full and partial observation settings, we take the initial distribution to have zero mean and covariance $\Sigma^{(0)} = 1.1\alpha \times I_{42}$. 
In the third column of Figure~\ref{fig:nonlinear_plot}, the dimension $d$ is varied over the values in $\{6, 18, 30, 42, 54, 66, 78, 90, 102\}$ which are all multiples of 3 to facilitate convenient calculations in the partially observed setting. We fix $N=20$ and $\alpha=10^{-4}$ and in both full and partial observation settings, we take the initial distribution to have zero mean and covariance $\Sigma^{(0)} = 1.1\alpha \times I_{d}$, respectively.

Our findings, as illustrated in Table~\ref{tab:nonlinear} and Figure~\ref{fig:nonlinear_plot}, demonstrate that $\mathsf{REnKF}$ achieves performance comparable to that of $\mathsf{EnKF}$, even in challenging nonlinear regimes. Notably, for both algorithms we observe a slightly inferior performance with partial observations compared to full observations under identical conditions.
Moreover, a consistent trend is noticed in the dependency of $\mathsf{E}{\text {\tiny L96}}$ on the noise level, state dimension, and ensemble size. Notice, however, that the performance of $\mathsf{REnKF}$ deteriorates further in non-Gaussian settings with partial observations, large $N,$ and large noise. Such worsened performance may be partly explained by the additional Gaussian assumption tacitly imposed in the resampling step, which further destroys the non-Gaussian structure of the problem for nonlinear forward models.  
Table~\ref{tab:nonlinear} further demonstrates that $\mathsf{REnKF}$ is as effective as $\mathsf{EnKF}$ in the task of uncertainty quantification. Nevertheless, both $\mathsf{EnKF}$ and $\mathsf{REnKF}$ encounter difficulties in delivering reliable uncertainty quantification, especially in scenarios with partial observation and small ensemble size.

\section{Proof of Theorem \ref{thm:MultiStepPOEnKFBounds}}\label{sec:proofmaintheorem}
The result will be established by strong induction on the \emph{mean bound} \eqref{eq:meanboundmainth} and the \emph{covariance bound} \eqref{eq:covarianceboundmainth} along with induction on two additional bounds: for any $j=1,2,\dots$ and $q \ge 1$
 \begin{align}
        \||\ttrace(\hatCpost^{(j-1)}) | \|_q 
        &\le c_3 r_2(\Cpost^{(0)}), \label{eq:covtracebound}\\
        \normn{|\hatCpr^{(j)}-\Cpr^{(j)}|}_{q}
        &\le 
        c_4
        \inparen{
        \sqrt{\frac{r_2(\Cpost^{(0)})}{N}}
        \lor 
        \sqrt{\frac{r_2(\xiCovariance)}{N}}
        \lor 
        \sqrt{\frac{r_2(\Gamma)}{N}}
        },   \label{eq:forecastcovbound}
    \end{align}
    where $c_3$ and $c_4$ are again potentially different universal constants that depend on the same parameters as $c_2$ in the statement of Theorem \ref{thm:MultiStepPOEnKFBounds}. We will refer to \eqref{eq:covtracebound} as the \emph{covariance trace bound} and to \eqref{eq:forecastcovbound} as the \emph{forecast covariance bound}. In this section, we require the following additional notation:
    given two positive sequences $\{a_n\}$ and $\{ b_n\}$, the relation $a_n \lesssim b_n$ denotes that $a_n \le c b_n$ for some constant $c>0$. If the constant $c$ depends on some quantity $\tau$, then we write $a \lesssim_{\tau} b$. Throughout, we denote positive universal constants by $c,c_1,c_2,c_3,c_4$, and the value of a universal constant may differ from line to line. In some cases, the explicit dependence of a universal constant on the parameter $\tau$ is indicated by writing $c(\tau)$.

    This section is organized as follows. 
    Subsection \ref{ssec:preliminaryresults} contains preliminary results. 
    We then prove the base case $j=1$ in Subsection \ref{ssec:base}. Finally, in Subsection \ref{ssec:induction} we show that the bounds \eqref{eq:meanboundmainth}, \eqref{eq:covarianceboundmainth}, \eqref{eq:covtracebound}, and \eqref{eq:forecastcovbound}  hold for $j$ assuming they hold for all $\ell \le j-1.$ 
\vspace{1mm}
\subsection{Preliminary Results}\label{ssec:preliminaryresults}
    \begin{lemma}[Operator Norm of Covariance] \label{lem:OpNormCovariance}
    For any $j \ge 0$, let $\Cpost^{(j)}$ be the analysis covariance at iteration $j$. Then, 
    \begin{align*}
        |\Cpost^{(j)}| 
        \le 
        |A|^{2 j} |\Cpost^{(0)}| + |\xiCovariance| \sum_{\ell=0}^{j-1} |A|^{2\ell}
        \le c(|A|, |\xiCovariance|, |\Cpost^{(0)}|, j).
    \end{align*}
\end{lemma}
\begin{proof}
    By Lemma~\ref{lem:TraceProperties}, $|\msC(\Cpr)| \le |\Cpr|,$ and so 
    \begin{align*}
        |\Cpost^{(j)}| 
        = |\msC(\Cpr^{(j)})|
        &\le |\Cpr^{(j)}|
        = |A \Cpost^{(j-1)} A^\top + \xiCovariance |
        \le |A|^2 |\Cpost^{(j-1)}| + |\xiCovariance|\\
        &\le |A|^4 |\Cpost^{(j-2)}| + |A|^2|\xiCovariance|+|\xiCovariance|
         \le \cdots \le 
        |A|^{2 j} |\Cpost^{(0)}| + |\xiCovariance| \sum_{\ell=0}^{j-1} |A|^{2\ell}.
    \end{align*}
\end{proof}

\begin{lemma}[Trace of Offset]\label{lem:OffsetTrace}
    For any $j \ge 1$, we have that 
    \begin{align*}
        \ttrace(\whatO^{(j)}) &\le 
        |H|^2 |\hatGamma^{(j)} - \Gamma| |\Gamma^{-1}|^2
        |\hatCpr^{(j)}| \ttrace(\hatCpr^{(j)}) \\
       & + 
        2(1+ |\hatCpr^{(j)}||H|^2|\Gamma^{-1}| )
        |\Gamma^{-1}|
        |\whatC^{(j)}_{\, u \eta}| |H| 
        \ttrace (
         \hatCpr^{(j)} 
         ).
    \end{align*}
\end{lemma}
\begin{proof}
    Write $\whatO^{(j)} = \sum_{\ell=1}^3\whatO^{(j)}_\ell$ with 
    \begin{align*}
        \whatO^{(j)}_1 &:= \msK(\hatCpr^{(j)}) (\widehat{\Gamma}^{(j)} - \Gamma) \msK^\top(\hatCpr^{(j)}), \\
         \whatO^{(j)}_2 &:= \bigl(I-\msK(\hatCpr^{(j)}) H\bigr) \whatC^{(j)}_{\, u \eta}\msK^\top(\hatCpr^{(j)}), \\
         \whatO^{(j)}_3 &:= \msK(\hatCpr^{(j)}) (\whatC^{(j)}_{\, u \eta})^\top \bigl(I-H^\top\msK^\top(\hatCpr^{(j)})\bigr).
    \end{align*}
    By linearity of the trace, $\ttrace(\whatO^{(j)}) = \sum_{\ell=1}^3\ttrace(\whatO^{(j)}_\ell)$. Note first that by Lemma~\ref{lem:TraceProperties} applied four times 
        \begin{align*}
        \ttrace(\hatO^{(j)}_1)
        &\le |\hatGamma^{(j)} - \Gamma| \ttrace(\msK^\top(\hatCpr^{(j)}) \msK(\hatCpr^{(j)}))\\
        &=|\hatGamma^{(j)} - \Gamma| 
        \ttrace((H\hatCpr^{(j)}H^\top + \Gamma)^{-1} H(\hatCpr^{(j)})^\top \hatCpr^{(j)} H^\top (H\hatCpr^{(j)}H^\top + \Gamma)^{-1})\\
        &\le|\hatGamma^{(j)} - \Gamma| |(H\hatCpr^{(j)}H^\top + \Gamma)^{-1}|^2
        \ttrace( (\hatCpr^{(j)})^\top \hatCpr^{(j)} H^\top H  )\\
        &\le|H|^2 |\hatGamma^{(j)} - \Gamma| |(H\hatCpr^{(j)}H^\top + \Gamma)^{-1}|^2
        |\hatCpr^{(j)}| \ttrace(\hatCpr^{(j)})\\
        &\le|H|^2 |\hatGamma^{(j)} - \Gamma| |\Gamma^{-1}|^2
        |\hatCpr^{(j)}| \ttrace(\hatCpr^{(j)}),
    \end{align*}
    where the final inequality holds since $H\hatCpr^{(j)}H^\top + \Gamma \succeq \Gamma$ implies that $\Gamma^{-1} \succeq (H\hatCpr^{(j)}H^\top + \Gamma )^{-1}$. Invoking once more Lemma~\ref{lem:TraceProperties} repeatedly, we get that 
    \begin{align*}
        \ttrace(\hatO^{(j)}_2)
        &= 
        \ttrace \bigl(
        \bigl(I-\msK(\hatCpr^{(j)}) H\bigr) \whatC^{(j)}_{\, u \eta} \msK^\top(\hatCpr^{(j)})
        \bigr)\\
        &\le 
        |\bigl(I-\msK(\hatCpr^{(j)}) H\bigr)|
        \ttrace \bigl(
         \whatC^{(j)}_{\, u \eta} \msK^\top(\hatCpr^{(j)})
        \bigr)\\
        &\le 
        |\bigl(I-\msK(\hatCpr^{(j)}) H\bigr)|
        |(H\hatCpr^{(j)}H^\top + \Gamma)^{-1}|
        \ttrace (
         H \hatCpr^{(j)} \whatC^{(j)}_{\, u \eta} 
         )\\
         &\le 
        |\bigl(I-\msK(\hatCpr^{(j)}) H\bigr)|
        |(H\hatCpr^{(j)}H^\top + \Gamma)^{-1}|
        |\whatC^{(j)}_{\, u \eta}| |H| 
        \ttrace (
         \hatCpr^{(j)} 
         )\\
         &\le 
        (1+ |\hatCpr^{(j)}||H|^2|\Gamma^{-1}| )
        |\Gamma^{-1}|
        |\whatC^{(j)}_{\, u \eta}| |H| 
        \ttrace (
         \hatCpr^{(j)} 
         ),
    \end{align*}
    where the last inequality holds since, by Lemma~\ref{lem:KalmanOpCtyBdd},
    \begin{align*}
        |\bigl(I-\msK(\hatCpr^{(j)}) H\bigr)|
        &\le 1+ |\msK(\hatCpr^{(j)})||H|
        \le 1+ |\hatCpr^{(j)}||H|^2|\Gamma^{-1}|.
    \end{align*}
    Finally, note that since $\hatO^{(j)}_3 = (\hatO^{(j)}_2)^\top,$ the analysis of $\hatO^{(j)}_3$ follows in similar fashion.
\end{proof}
        
    \subsection{Base Case}\label{ssec:base}
    In the next four subsections we establish the covariance trace bound \eqref{eq:covtracebound}, the forecast covariance bound \eqref{eq:forecastcovbound}, the mean bound \eqref{eq:meanboundmainth}, and the covariance bound \eqref{eq:covarianceboundmainth} in the base case $j = 1.$
    \subsubsection{Covariance Trace Bound} 
    Since  $\hatCpost^{(0)} = \Cpost^{(0)},$ we directly obtain that
    \begin{align*}
        \| |\ttrace(\hatCpost^{(0)})| \|_q = \ttrace( \Cpost^{(0)}) =   |\Cpost^{(0)}| r_2(\Cpost^{(0)}).
    \end{align*}

    \subsubsection{Forecast Covariance Bound}\label{sssec:forecastbase} 
    Let $\mfq \in \{q, 2q, 4q\}$. It follows by the triangle inequality, Theorem~\ref{thm:CovarianceExpectationBoundLq}, and Lemma~\ref{lem:CrossCovarianceBound} that 
    \begin{align} \label{eq:BaseCasePOForecastCovarianceBound}  
        \normn{|\hatCpr^{(1)}-\Cpr^{(1)}|}_{\mfq}
        &\le |A|^2 \| | S^{(0)} - \Cpost^{(0)}| \|_\mfq + 
        \| |\hatxiCovariance^{(1)} - \xiCovariance| \|_\mfq+ 2|A| \| |\whatC_{\, u \xi}^{(1)}| \|_{\mfq} \nonumber\\
        &\lesssim_q |A|^2 |\Cpost^{(0)}| \sqrt{\frac{r_2(\Sigma^{(0)})}{N}} 
        + |\xiCovariance| \sqrt{\frac{r_2(\xiCovariance)}{N}}\nonumber
        \\&+ 2|A| (|\Cpost^{(0)}| \lor |\xiCovariance|) \inparen{
    \sqrt{\frac{r_2(\Cpost^{(0)})}{N}}
    \lor 
    \sqrt{\frac{r_2(\xiCovariance)}{N}}} \nonumber\\
    &\le c(|A|, |\Sigma^{(0)}|, |\xiCovariance|, q) 
        \inparen{ \sqrt{\frac{r_2(\Cpost^{(0)})}{N}} \lor \sqrt{\frac{r_2(\xiCovariance)}{N}}}.
    \end{align}
    \subsubsection{Mean Bound} By Lemma~\ref{lem:MeanOpCtyBddLp}, 
\begin{align*}
    \||\hatmpost^{(1)} - \mpost^{(1)}| \|_q
    &= \|| \msM(\hatmpr^{(1)}, \hatCpr^{(1)}; y^{(1)}) - \msM(\mpr^{(1)}, \Cpr^{(1)}; y^{(1)}) |\|_q + \| | \msK(\hatCpr^{(1)}) \bareta^{(1)}|\|_q \\
    &\le 
        \norm{|\hatmpr^{(1)}-\mpr^{(1)}|}_{q}
        + 
        |H|^2 |\Gamma^{-1}| \normn{|\hatCpr^{(1)}|}_{2q} \normn{|\hatmpr^{(1)}-\mpr^{(1)}|}_{2q}\\  
        & + 
        \normn{|\hatCpr^{(1)}-\Cpr^{(1)}|}_{q} |H| |\Gamma^{-1}|
        \bigl(1 + |H|^2 |\Gamma^{-1}| |\Cpr^{(1)}| \bigr) 
        |y^{(1)}-H\mpr^{(1)}|\\
        &+ \| | \msK(\hatCpr^{(1)}) \bareta^{(1)}|\|_q.
\end{align*}
The mean bound \eqref{eq:meanboundmainth} with $j =1$ is then a direct consequence of the bounds that we now establish on $\||\hatmpr^{(1)} - \mpr^{(1)}| \|_\mfq$ for $\mfq \in \{ q, 2q\}$ and on $\| | \msK(\hatCpr^{(1)}) \bareta^{(1)}|\|_q$.

\paragraph{Controlling $\||\hatmpr^{(1)} - \mpr^{(1)}| \|_\mfq$ for $\mfq \in \{ q, 2q\}$} It follows by the triangle inequality and Lemma~\ref{lem:GaussianTraceLqBound} applied twice that 
    \begin{align*}
        \||\hatmpr^{(1)} - \mpr^{(1)}| \|_\mfq
        &\le |A| \| | \bar{\upr}^{(0)} - \mpost^{(0)}| \|_\mfq + \| | \barxi^{(1)}| \|_\mfq
        \lesssim_q |A| |\Cpost^{(0)}|  \sqrt{\frac{r_2(\Cpost^{(0)})}{N}}
        + |\xiCovariance| \sqrt{\frac{r_2(\xiCovariance)}{N}}\\
        &\le c(|A|, |\Sigma^{(0)}|, |\xiCovariance|, q) 
        \inparen{\sqrt{\frac{r_2(\Cpost^{(0)})}{N}} \lor \sqrt{\frac{r_2(\xiCovariance)}{N}}}.
    \end{align*}

    \paragraph{Controlling $\| | \msK(\hatCpr^{(1)}) \bareta^{(1)}|\|_q$} By Cauchy-Schwarz
    $$
    \| | \msK(\hatCpr^{(1)}) \bareta^{(1)}|\|_q \le \| | \msK(\hatCpr^{(1)})|\|_{2q}  \| | \bareta^{(1)} |\|_{2q}.$$ 
    We bound each term in turn. By Lemma~\ref{lem:GaussianTraceLqBound}, $\| | \bareta^{(1)} |\|_{2q} \lesssim_q \sqrt{\ttrace(\Gamma)/N}=\sqrt{|\Gamma|r_2(\Gamma)/N},$
    and by Lemma~\ref{lem:KalmanOpCtyBdd} and the forecast covariance bound \eqref{eq:BaseCasePOForecastCovarianceBound},  
    \begin{align*}
        \| | \msK(\hatCpr^{(1)})|\|_{2q} 
        &\le |H| |\Gamma^{-1}| \| |\hatCpr^{(1)} |\|_{2q}
        \le |H| |\Gamma^{-1}| \inparen {
        \| |\hatCpr^{(1)} -\Cpr^{(1)}|\|_{2q} + |\Cpr^{(1)}|
        }\\
        &\lesssim 
        |H| |\Gamma^{-1}| |\Cpr^{(1)}|
        c(|A|, |\Sigma^{(0)}|, |\xiCovariance|, q) 
        \inparen{1 \lor \sqrt{\frac{r_2(\Cpost^{(0)})}{N}} \lor \sqrt{\frac{r_2(\xiCovariance)}{N}}}.
    \end{align*}
    Therefore, 
    \begin{align*}
        \| | \msK(\hatCpr^{(1)}) \bareta^{(1)}|\|_q 
        &\le
        c(|H|, |\Sigma^{(0)}|, |\xiCovariance|,|H|, |\Gamma^{-1}|,|\Gamma|, q) 
        \inparen{\sqrt{\frac{r_2(\Cpost^{(0)})}{N}} 
        \lor \sqrt{\frac{r_2(\xiCovariance)}{N}}
        \lor \sqrt{\frac{r_2(\Gamma)}{N}}}.
    \end{align*}

\subsubsection{Covariance Bound}
 From Lemma~\ref{lem:CovarOpCtyBddLp} and the forecast covariance bound derived in Subsection \ref{sssec:forecastbase},  
 we have 
    \begin{align*}
        \normn{ | \hatCpost^{(1)} - \Cpost^{(1)}| }_{q}
        &\le 
        \normn{ | \msC(\hatCpr^{(1)}) - \msC(\Cpr^{(1)})| }_{q} + \normn{ | \hatO| }_{q}\\
        &\le
        \normn{| \hatCpr^{(1)} - \Cpr^{(1)}|}_{q}
        (1 + |H|^2 |\Gamma^{-1}| |\Cpr^{(1)}|)\\
        &+ (
            |H|^2 |\Gamma^{-1}| + 
            |H|^4 |\Gamma^{-1}|^2 |\Cpr^{(1)}|
        ) 
        \normn{|\hatCpr^{(1)}|}_{2q} \normn{|\hatCpr^{(1)}-\Cpr^{(1)}|}_{2q}
        + \normn{ | \hatO^{(1)}| }_{q}\\
        &\le
        c(|A|, |H|, |\Gamma^{-1}|, |\Cpost^{(0)}|, |\xiCovariance|, q)
        \inparen{
        \sqrt{\frac{r_2(\Cpost^{(0)})}{N}}
        \lor 
        \sqrt{\frac{r_2(\xiCovariance)}{N}}
        }+ \normn{ | \hatO^{(1)}| }_{q}.
    \end{align*}
    To derive the covariance bound \eqref{eq:covarianceboundmainth}, we need to control the offset term $  \normn{|\whatO^{(1)}|}_q$.
     First, using the triangle inequality we write 
    \begin{align*}
        \normn{|\whatO^{(1)}|}_q &\le  
        \normn{ |\msK(\hatCpr^{(1)}) (\widehat{\Gamma}^{(1)} - \Gamma) \msK^\top(\hatCpr^{(1)})| }_q 
        +\normn{| \bigl(I-\msK(\hatCpr^{(1)}) H\bigr) \whatC^{(1)}_{\, u \eta} \msK^\top(\hatCpr^{(1)})| }_q \\
        &+\normn{ | \msK(\hatCpr^{(1)}) (\whatC^{(1)}_{\, u \eta})^\top \bigl(I-H^\top\msK^\top(\hatCpr^{(1)})\bigr)| }_q\\
        &=: 
        \| |\hatO^{(1)}_1 |\|_q +\| |\hatO^{(1)}_2 |\|_q+ \| |\hatO^{(1)}_3 |\|_q.
    \end{align*}
    We next bound each term in turn. 
    
      \paragraph{Controlling $\| |\hatO^{(1)}_1 |\|_q$} By Lemma~\ref{lem:KalmanOpCtyBdd}, Theorem~\ref{thm:CovarianceExpectationBoundLq}, and the forecast covariance bound \eqref{eq:BaseCasePOForecastCovarianceBound},
      it holds that 
 \begin{align*}
     \normn{ | \msK(\hatCpr^{(1)})  (\widehat{\Gamma}^{(1)} - \Gamma) & \msK^\top(\hatCpr^{(1)})| }_q
     \le
     \normn{ |\msK(\hatCpr^{(1)})| }^2_{4q} \normn{| \widehat{\Gamma}^{(1)} - \Gamma| }_{2q}
     \\
     &\le
     |H|^2 |\Gamma^{-1}|^2 \normn{ |\hatCpr^{(1)}| }^2_{4q} \normn{| \widehat{\Gamma}^{(1)} - \Gamma| }_{2q}\\
     &\lesssim_q 
     |H|^2 |\Gamma^{-1}|^2
     |\Gamma | \sqrt{\frac{r_2(\Gamma)}{N}}
      \inparen{ 
        1 \lor 
        \sqrt{\frac{r_2(\Cpost^{(0)})}{N}}
        \lor
        \sqrt{\frac{r_2(\xiCovariance)}{N}}
        \lor 
        \sqrt{\frac{r_2(\Gamma)}{N}}
        }^2\\
     &\le c(|H|, |\Gamma|, |\Gamma^{-1}|,|\xiCovariance|, |\Cpost^{(0)}|, q ) \sqrt{\frac{r_2(\Gamma)}{N}},
 \end{align*}
 where the last inequality uses the fact that $N \ge r_2(\Cpost^{(0)}) \lor r_2(\xiCovariance) \lor r_2(\Gamma)$. 
    
    \paragraph{Controlling $\| |\hatO^{(1)}_2 |\|_q$} By Lemma~\ref{lem:KalmanOpCtyBdd}, Lemma~\ref{lem:CrossCovarianceBound}, and the forecast covariance bound \eqref{eq:BaseCasePOForecastCovarianceBound}, we get
 \begin{align*}
     \normn{ 
     | \bigl(I-\msK(\hatCpr^{(1)}) & H\bigr)   \whatC^{(1)}_{\, u \eta}  \msK^\top(\hatCpr^{(1)})|
     }_q
     \le 
    \normn{
    |\msK(\hatCpr^{(1)})| 
    |I-\msK(\hatCpr^{(1)}) H| 
    |\whatC^{(1)}_{\, u \eta}| }_q \\
     &\le 
     \normn{
     |\msK(\hatCpr^{(1)})| 
     (1+|\msK(\hatCpr^{(1)})| |H|)  
     |\whatC^{(1)}_{\, u \eta}| }_q\\
     &\le 
     \normn{|\whatC^{(1)}_{\, u \eta}| }_{2q}
     \inparen{
     |H| |\Gamma^{-1}|  \| |\hatCpr^{(1)}|\|_{2q}
     + |H|^3 |\Gamma^{-1}|^2  \| |\hatCpr^{(1)}|^2\|_{2q}
     }\\
     &\le
     c(|A|, |\Sigma^{(0)}|, |\xiCovariance|, q)
    \inparen{
    1 
    \lor 
    \sqrt{\frac{r_2(\Cpost^{(0)})}{N}}
    \lor 
    \sqrt{\frac{r_2(\xiCovariance)}{N}}
    \lor 
    \sqrt{\frac{r_2(\Gamma)}{N}}
    }\\
    &\times \inparen{
    |H|^3 |\Gamma^{-1}|  
     + |H|^7 |\Gamma^{-1}|^2  
    } (|\Cpost^{(0)}| \lor |\Gamma|) 
    \inparen{
        \sqrt{\frac{r_2(\Cpost^{(0)})}{N}}
        \lor
        \sqrt{\frac{r_2(\xiCovariance)}{N}}
        \lor 
        \sqrt{\frac{r_2(\Gamma)}{N}}
        }\\
        &\le c( |A|, |H|, |\Gamma^{-1}|, |\Gamma|, |\Cpost^{(0)}|, |\xiCovariance|, q) \inparen{
        \sqrt{\frac{r_2(\Cpost^{(0)})}{N}}
        \lor
        \sqrt{\frac{r_2(\xiCovariance)}{N}}
        \lor 
        \sqrt{\frac{r_2(\Gamma)}{N}}
        }.
 \end{align*}
    \paragraph{Controlling $\| |\hatO^{(1)}_3 |\|_q$} Note that $ \| |\hatO^{(1)}_3 |\|_q = \| |\hatO^{(1)}_2 |\|_q.$

\subsection{Induction Step}\label{ssec:induction}
In this subsection, to reduce notation we write $\Omega = 
        \sqrt{\frac{r_2(\Cpost^{(0)})}{N}}
        \lor 
        \sqrt{\frac{r_2(\xiCovariance)}{N}}
        \lor 
        \sqrt{\frac{r_2(\Gamma)}{N}}
        $.
Throughout, we work under the \emph{inductive hypothesis} that, for all $\ell \le j -1,$ it holds that 
\begin{align}
\begin{alignedat}{2} \label{eq:inductivehypothesis}
    \||\ttrace(\hatCpost^{(\ell-1)}) | \|_q 
        &\le 
        c_1 r_2(\Cpost^{(0)}), 
        \qquad 
        &\normn{|\hatCpr^{(\ell)}-\Cpr^{(\ell)}|}_{\mfq}
        \le 
        c_2 \Omega, \\
        \||\hatmpost^{(\ell)} - \mpost^{(\ell)}| \|_q
        &\le     
        c_3 \Omega,
        &\normn{ | \hatCpost^{(\ell)} - \Cpost^{(\ell)}| }_{q}
        \le 
        c_4 \Omega,
\end{alignedat}
\end{align}
    where $c_1, c_2, c_3,$ and  $c_4$ are constants depending on
    $|\Cpost^{(0)}|, |A|, |H|, |\Gamma^{-1}|, |\Gamma|, |\xiCovariance|, q$ and $ j,$
    and $c_3$ additionally depends on  $\{|y^{(i)}-H \mpr^{(i)}| \}_{i \le \ell-1}$. For the remainder of the proof, $c$ and $c'$ denote constants that depend on  $|\Cpost^{(0)}|, |A|, |H|, |\Gamma^{-1}|, |\Gamma|, |\xiCovariance|, q$ and $ j,$
    and $c'$ additionally depends on  $\{|y^{(i)}-H \mpr^{(i)}| \}_{i \le \ell-1}$ and are potentially different from line to line.   
    In the next four subsections we show that, under the inductive hypothesis, the four bounds in $\eqref{eq:inductivehypothesis}$ also hold for $\ell = j.$
    Throughout, we use without further notice that  $|\Cpost^{(\ell)}| \lesssim c(|A|,|\xiCovariance|, |\Cpost^{(0)}|, \ell),$ which was proved in Lemma~\ref{lem:OpNormCovariance}.

    \subsubsection{Covariance Trace Bound}\label{sssec:covtracebound} By Lemma~\ref{lem:CovarOpCtyBddLp}, $\ttrace(\msC(\hatCpr^{(j-1)})) \le \ttrace(\hatCpr^{(j-1)})$ follows from the fact that $\msC(\hatCpr^{(j-1)}) \preceq \hatCpr^{(j-1)}$, and so
    \begin{align*}
        \| \ttrace(\hatCpost^{(j-1)}) \|_{q}
        &\le 
        \| \ttrace(\msC(\hatCpr^{(j-1)})) \|_{q}  + \| \ttrace(\hatO^{(j-1)})\|_{q}
        \le
        \| \ttrace(\hatCpr^{(j-1)})\|_{q}  + \| \ttrace(\hatO^{(j-1)})\|_{q},
    \end{align*}
    We will show that both of the terms on the right-hand side are bounded above by a constant times $r_2(\Cpost^{(0)})$.

        \paragraph{Controlling $\| \ttrace(\hatCpr^{(j-1)})\|_{q}$}
        Noting first that
        \begin{align*}
            \E \insquare{\hatCpr^{(j-1)} \Big| \hatmpost^{(j-2)}, \hatCpost^{(j-2)} } 
            &=\E \insquare{A S^{(j-2)} A^\top 
        + \hatxiCovariance^{(j-1)} + A \whatC^{(j-1)}_{\, u \xi} + \whatC^{(j-1)}_{\, \xi u}A^\top \Big| \hatmpost^{(j-2)}, \hatCpost^{(j-2)}}\\ 
        &= \E \insquare{A S^{(j-2)} A^\top \Big| \hatmpost^{(j-2)}, \hatCpost^{(j-2)}} = A\hatCpost^{(j-2)}A^\top,
        \end{align*}
         and by Lemma~\ref{lem:TraceProperties}, it holds almost surely that 
        \begin{align*}
            \ttrace \bigl(A\hatCpost^{(j-2)}A^\top \bigr)
            \le  |A|^2 \ttrace(\hatCpost^{(j-2)})
            =|A|^2 |\hatCpost^{(j-2)}| r_2(\hatCpost^{(j-2)}).
        \end{align*}
        Then, by iterated expectations and Lemma~\ref{lem:TraceConcentration2}, we have 
        \begin{align*}
            \E \insquare{\ttrace(\hatCpr^{(j-1)}) }^q
            &=
            \E \insquare{ 
            \E \insquare{
                \inparen{\ttrace(\hatCpr^{(j-1)})}^q
               \, \Big| \,
                \hatmpost^{(j-2)}, \hatCpost^{(j-2)}
                }
            }\\
            & \hspace{-0.55cm} \lesssim_q
            \E \insquare{ 
            \E \insquare{
                \inparen{
                \ttrace \inparen{
                \hatCpr^{(j-1)}
                -
                \E \bigl[\hatCpr^{(j-1)} \big| \hatmpost^{(j-2)}, \hatCpost^{(j-2)} \bigr]  
               } }^q \,
                \Big| \, 
                \hatmpost^{(j-2)}, \hatCpost^{(j-2)}
                }
            }\\
            & \hspace{6.27cm}+
            \E \insquare{ 
                \inparen{
                \ttrace\inparen{
                \E \bigl[\hatCpr^{(j-1)} \big| \hatmpost^{(j-2)}, \hatCpost^{(j-2)} \bigr]  
                }  }^q
            }\\
           &  \hspace{-0.55cm}\lesssim
            \E \insquare{ 
            \inparen{
                \frac{ 
                \ttrace \bigl(
                \E \bigl[ \hatCpr^{(j-1)} \big| \hatmpost^{(j-2)}, \hatCpost^{(j-2)} \bigr] \bigr)}{\sqrt{N}}  }^{q}
            }
            +
            \E \insquare{ 
                \inparen{
                \ttrace \inparen{
                \E \insquare{\hatCpr^{(j-1)}| \hatmpost^{(j-2)}, \hatCpost^{(j-2)} } } }^q
            }\\
            & \hspace{-0.55cm} \le
            \frac{|A|^{2q}}{N^{q/2}} 
            \E \insquare{\inparen{ \ttrace(\hatCpost^{(j-2)})}^{q}}
            +
            |A|^{2q} \E \insquare{ \inparen{\ttrace(\hatCpost^{(j-2)})}^{ q}}\\
            & \hspace{-0.55cm} \lesssim
            \frac{|A|^{2q}}{N^{q/2}} 
            c r_2(\Cpost^{(0)})^q
            +
            |A|^{2q} c r_2(\Cpost^{(0)})^q
            \le c r_2(\Cpost^{(0)})^q,
            \end{align*}
            where the second to last inequality holds by the inductive hypothesis \eqref{eq:inductivehypothesis}.
        
        \paragraph{Controlling $\|\hatO^{(j-1)}\|_q$} By definition, we have
    \begin{align*}
        \whatO^{(j-1)} &=
        \msK(\hatCpr^{(j-1)}) (\widehat{\Gamma}^{(j-1)} - \Gamma) \msK^\top(\hatCpr^{(j-1)}) 
        + \bigl(I-\msK(\hatCpr^{(j-1)}) H\bigr) \whatC^{(j-1)}_{\, u \eta} \msK^\top(\hatCpr^{(j-1)}) \\
        &+ \msK(\hatCpr^{(j-1)}) (\whatC^{(j-1)}_{\, u \eta})^\top \bigl(I-H^\top\msK^\top(\hatCpr^{(j-1)}) \bigr)\\
        &=: \hatO^{(j-1)}_1+\hatO^{(j-1)}_2+\hatO^{(j-1)}_3.
    \end{align*}
    Therefore,  $\| \ttrace(\hatO^{(j-1)})\|_{q}
        \le 
        \| \ttrace(\hatO^{(j-1)}_1)\|_{q}
        +
        \| \ttrace(\hatO^{(j-1)}_2)\|_{q}
        +
        \| \ttrace(\hatO^{(j-1)}_3)\|_{q}$.

        \paragraph{Controlling $\| \ttrace(\hatO^{(j-1)}_1)\|_{q}$} By Lemma~\ref{lem:OffsetTrace},
        \begin{align*}
        \ttrace(\hatO^{(j-1)}_1)
        &\le|H|^2 |\hatGamma^{(j-1)} - \Gamma| |\Gamma^{-1}|^2
        |\hatCpr^{(j-1)}| \ttrace(\hatCpr^{(j-1)}),
    \end{align*}
    and so 
    \begin{align*}
        \| \ttrace(\hatO^{(j-1)}_1)\|_{q} 
        &\le 
        |H|^2 |\Gamma^{-1}|^2 \| |\hatGamma^{(j-1)} - \Gamma|  |\hatCpr^{(j-1)}| \ttrace(\hatCpr^{(j-1)}) \|_q\\
        &\le 
        |H|^2 |\Gamma^{-1}|^2 \| |\hatGamma^{(j-1)} - \Gamma| \|_{2q} \||\hatCpr^{(j-1)}|\|\|_{4q} \|\ttrace(\hatCpr^{(j-1)}) \|_{4q}.
    \end{align*}
    By Theorem~\ref{thm:CovarianceExpectationBoundLq}, $\| |\hatGamma^{(j-1)} - \Gamma| \|_{2q} \lesssim_q |\Gamma| \sqrt{\frac{r_2(\Gamma)}{N}},$ and by the inductive hypothesis \eqref{eq:inductivehypothesis} and the fact that $|\Cpr^{(j-1)}| \le |A|^2 |\Cpost^{(j-2)}| + |\xiCovariance|$, we have 
    \begin{align*}
        \||\hatCpr^{(j-1)}|\|_{4q} 
        &\le |\Cpr^{(j-1)}| + \||\hatCpr^{(j-1)} - \Cpr^{(j-1)}|\|_{4q}
        \le 
        c
        \inparen{1 \lor \Omega}.
    \end{align*}
    We have also previously shown that $\|\ttrace(\hatCpr^{(j-1)}) \|_{4q} \lesssim r_2(\Cpost^{(0)})$. Noting that $\inparen{1 \lor \Omega} r_2(\Cpost^{(0)})\lesssim r_2(\Cpost^{(0)}), $ we get that $\| \ttrace(\hatO^{(j-1)}_1)\|_{q} \le c r_2(\Cpost^{(0)}).$
        
    \paragraph{Controlling $\||\ttrace(\hatO^{(j-1)}_2)|\|_q$} By Lemma~\ref{lem:OffsetTrace},
     \begin{align*}
        \ttrace(\hatO^{(j-1)}_2)
        &\le 
        (1+ |\hatCpr^{(j-1)}||H|^2|\Gamma^{-1}| )
        |\Gamma^{-1}||H|
        |\whatC^{(j-1)}_{\, u \eta}|  
        \ttrace (
         \hatCpr^{(j-1)} 
         ).
    \end{align*}
    Therefore, 
    \begin{align*}
        \| |\ttrace(\hatO^{(j-1)}_2)| \|_q 
        &\le 
        (1+ \| |\hatCpr^{(j-1)}|\|_{2q} |H|^2|\Gamma^{-1}| )
        |\Gamma^{-1}||H|
        \||\whatC^{(j-1)}_{\, u \eta}|  \|_{4q}
        \| \ttrace (\hatCpr^{(j-1)} )\|_{4q}.
    \end{align*}
    By iterated expectation, Lemma~\ref{lem:CrossCovarianceBound}, and Lemma~\ref{lem:TraceConcentration2} we have, for $\mfq \in \{q,2q,4q\},$
    \begin{align} \label{eq:crosscovarianceInductive}
        \normn{|\whatC^{(j-1)}_{\, u \eta}| }_{\mfq}
        &= \E \insquare{\E \insquare{ |\whatC^{(j-1)}_{\, u \eta}|^\mfq | \hatmpost^{(j-2)}, \hatCpost^{(j-2)} }}^{1/\mfq} \nonumber\\
        &\lesssim_q 
        \norm{
        (|\hatCpost^{(j-2)}| \lor |\Gamma|) \inparen{
        \sqrt{\frac{r_2(\hatCpost^{(j-2)})}{N}}
        \lor 
        \sqrt{\frac{r_2(\Gamma)}{N}}} 
        }_\mfq \nonumber \\
        &\le
        (\normn{|\hatCpost^{(j-2)}|}_{2 \mfq} \lor |\Gamma|) 
        \inparen{
        \norm{
        \sqrt{\frac{r_2(\hatCpost^{(j-2)})}{N}}
        }_{2\mfq}
        \lor 
        \sqrt{\frac{r_2(\Gamma)}{N}}
        }.
    \end{align}
    By the triangle inequality and the inductive hypothesis \eqref{eq:inductivehypothesis}, it follows that
    \begin{align*}
        \normn{|\hatCpost^{(j-2)}|}_{2 \mfq}  \le \normn{|\hatCpost^{(j-2)} - \Cpost^{(j-2)}|}_{2 \mfq}  + |\Cpost^{(j-2)}| 
        &\le  c \inparen{
        1 \lor \Omega},
    \end{align*}
    and also that 
    \begin{align*}
        \norm{
        \sqrt{\frac{r_2(\hatCpost^{(j-2)})}{N}}
        }_{2\mfq}^{2\mfq}
        =
        \E \insquare{
        \inparen{\frac{r_2(\hatCpost^{(j-2)})}{N}}^{\mfq} }
        \lesssim
        \E \insquare{
        \inparen{\frac{\ttrace(\hatCpost^{(j-2)})}{N}}^{\mfq}}
        \le
        cN^{-\mfq} r_2(\Cpost^{(0)})^{\mfq}.
    \end{align*}
Using identical arguments to those used to control $\| |\ttrace(\hatO^{(j-1)}_1)| \|_q$, we  have that 
    \begin{align*}
        \| \ttrace(\hatO^{(j-1)}_2)\|_{q} 
        &\le
        c r_2(\Cpost^{(0)}).
    \end{align*}

    \paragraph{Controlling $\| \ttrace(\hatO^{(j-1)}_3) \|_{q}$} Note that $\| \ttrace(\hatO^{(j-1)}_2)\|_{q} = \| \ttrace(\hatO^{(j-1)}_2)\|_{q}.$

    \subsubsection{Forecast Covariance Bound}\label{sssec:forecastinduction}
    Let $\mfq \in \{q,2q,4q\}$. By the triangle inequality, the inductive hypothesis \eqref{eq:inductivehypothesis}, and Theorem~\ref{thm:CovarianceExpectationBoundLq}, we have 

\begin{align*}
        \normn{|\hatCpr^{(j)}- &\Cpr^{(j)}|}_{\mfq}
        \le |A|^2 \| | S^{(j-1)} - \Cpost^{(j-1)}| \|_\mfq + 
        \| |\hatxiCovariance^{(j)} - \xiCovariance| \|_\mfq+ 2|A| \| |\whatC_{\, u \xi}^{(j)}| \|_{\mfq}\\
        &\le |A|^2 
        \inparen{
        \| | S^{(j-1)} - \hatCpost^{(j-1)} | \|_\mfq
        +
        \|| \hatCpost^{(j-1)} - \Cpost^{(j-1)}|\|_\mfq
        }
        + \| |\hatxiCovariance^{(j)}  - \xiCovariance| \|_\mfq 
        + 2|A| \| |\whatC_{\, u \xi}^{(j)}| \|_{\mfq}\\
        &\le c|A|^2 
        \inparen{
        \| | S^{(j-1)} - \hatCpost^{(j-1)} | \|_\mfq
        +
        \Omega}
        +|\xiCovariance| \inparen{1 \lor \sqrt{\frac{r_2(\xiCovariance)}{N}}}
        + 2|A| \| |\whatC_{\, u \xi}^{(j)}| \|_{\mfq}.
    \end{align*}
The forecast covariance bound \eqref{eq:forecastcovbound} is then a direct consequence of the bounds that we now establish on $\normn{|\whatC^{(j)}_{\, u \xi}| }_{\mfq}$ and $\| | S^{(j-1)} - \hatCpost^{(j-1)} | \|_\mfq.$

    \paragraph{Controlling $\normn{|\whatC^{(j)}_{\, u \xi}| }_{\mfq}$} By an identical analysis to the one used in bounding $\normn{|\whatC^{(j)}_{\, u \eta}| }_{\mfq}$ in \eqref{eq:crosscovarianceInductive}, we have that 
    \begin{align}\label{eq:CrossCovarianceBoundInductiveStep}
        \normn{|\whatC^{(j)}_{\, u \xi}| }_{\mfq}
        &\lesssim 
        c\Omega.
    \end{align}

    \paragraph{Controlling $\| | S^{(j-1)} - \hatCpost^{(j-1)} | \|_\mfq$} By iterated expectations and Theorem~\ref{thm:CovarianceExpectationBoundLq}, we have
    
    \begin{align*}
    \normn{|S^{(j-1)}-\hatCpost^{(j-1)}|}_{\mfq}^\mfq
    &= \E \insquare{|S^{(j-1)}-\hatCpost^{(j-1)}|^\mfq}
    = \E \insquare{
    \E\insquare{
    |S^{(j-1)} - \hatCpost^{(j-1)}|^\mfq \bigg | \hatmpost^{(j-1)}, \hatCpost^{(j-1)}} } \\
    & \lesssim_q
    \E \insquare{
    |\hatCpost^{(j-1)}|^\mfq \inparen{\frac{r_2(\hatCpost^{(j-1)})}{N}
    }^{\mfq/2}}
    =
    \E \insquare{
    |\hatCpost^{(j-1)}|^{\mfq/2} \inparen{\frac{\ttrace(\hatCpost^{(j-1)})}{N}
    }^{\mfq/2}}\\
    &\le
    \sqrt{ 
    \E |\hatCpost^{(j-1)}|^{\mfq}
    } 
    \sqrt{ \E \insquare{\frac{\ttrace(\hatCpost^{(j-1)})}{N}}^{\mfq}} \,.
\end{align*}

By the covariance trace bound proved in Subsection \ref{sssec:covtracebound} and the inductive hypothesis \eqref{eq:inductivehypothesis}, we then have that $\normn{|S^{(j-1)}-\hatCpost^{(j-1)}|}_{\mfq} \lesssim c \Omega.$

    \subsubsection{Mean Bound} By Lemma~\ref{lem:MeanOpCtyBddLp}, we have
\begin{align*}
    \||\hatmpost^{(j)} - \mpost^{(j)}| \|_q
    &= \|| \msM(\hatmpr^{(j)}, \hatCpr^{(j)}; y^{(j)}) - \msM(\mpr^{(j)}, \Cpr^{(j)}; y^{(j)}) |\|_q + \| |\msK(\hatCpr^{(j)}) \bareta^{(j)} |  \|_q\\
    &\le 
        \norm{|\hatmpr^{(j)}-\mpr^{(j)}|}_{q}
        + 
        |H|^2 |\Gamma^{-1}| \normn{|\hatCpr^{(j)}|}_{2q} \normn{|\hatmpr^{(j)}-\mpr^{(j)}|}_{2q}\\  
        & + 
        \normn{|\hatCpr^{(j)}-\Cpr^{(j)}|}_{q} |H| |\Gamma^{-1}|
        \bigl(1 + |H|^2 |\Gamma^{-1}| |\Cpr^{(j)}| \bigr) 
        |y^{(j)}-H\mpr^{(j)}|\\
        &+ \| |\msK(\hatCpr^{(j)}) \bareta^{(j)} |  \|_q.
\end{align*}
The induction step for the mean bound \eqref{eq:meanboundmainth} is then a direct consequence of the bounds that we now establish on $\||\hatmpr^{(j)} - \mpr^{(j)}| \|_\mfq$ for $\mfq \in \{ q, 2q\}$ and on $\| | \msK(\hatCpr^{(j)}) \bareta^{(j)}|\|_q$.
    \paragraph{Controlling $\||\hatmpr^{(j)} - \mpr^{(j)}| \|_\mfq$ for $\mfq \in \{ q, 2q\}$} It follows by the triangle inequality and the inductive hypothesis \eqref{eq:inductivehypothesis} that
    \begin{align*}
        \||\hatmpr^{(j)} - \mpr^{(j)}| \|_\mfq
        &\le |A| \| | \bar{\upr}^{(j-1)} - \mpost^{(j-1)}| \|_\mfq + \| | \barxi^{(j)}| \|_\mfq\\
        &\le |A| 
        \inparen{
        \| | \bar{\upr}^{(j-1)} - \hatmpost^{(j-1)} | \|_\mfq
        +
        \|| \hatmpost^{(j-1)} - \mpost^{(j-1)}|\|_\mfq
        }
        + \| | \barxi^{(j)}| \|_\mfq
        \\
        &\le c' |A| 
        \inparen{
        \| | \bar{\upr}^{(j-1)} - \hatmpost^{(j-1)} | \|_\mfq
        +
        \Omega 
        }+ 
        c(q) |\xiCovariance| \sqrt{\frac{r_2(\xiCovariance)}{N}}.
    \end{align*}

By iterated expectations, Lemma~\ref{lem:GaussianTraceLqBound}, and the covariance trace bound proved in Subsection \ref{sssec:covtracebound}, it follows that
\begin{align*}
    \| | \bar{\upr}^{(j-1)} - \hatmpost^{(j-1)} | \|_\mfq^\mfq
    &= \E \bigl[ | \bar{\upr}^{(j-1)} - \hatmpost^{(j-1)} |^\mfq \bigr]
    = \E \insquare{\E  \insquare{ | \bar{\upr}^{(j-1)} - \hatmpost^{(j-1)} |^\mfq  \, \Big| \, \hatmpost^{(j-1)}, \hatCpost^{(j-1)} } }\nonumber \\
    &\lesssim_q 
    \E \insquare{\inparen{\frac{\ttrace(\hatCpost^{(j-1)})}{N}}^{\mfq/2}}
    \le c \inparen{\frac{r_2(\Cpost^{(0)})}{N}}^{\mfq/2},
\end{align*}
and so $\normn{|\hatmpr^{(j)}-\mpr^{(j)}|}_{\mfq} \lesssim 
 c'\Omega.$

\paragraph{Controlling $\| | \msK(\hatCpr^{(j)}) \bareta^{(j)}|\|_q$} Note first that $\| | \msK(\hatCpr^{(j)}) \bareta^{(j)}|\|_q \le
        \| | \msK(\hatCpr^{(j)})|\|_{2q}  \| | \bareta^{(j)} |\|_{2q}$.
    We then have by Lemma~\ref{lem:GaussianTraceLqBound}, $\| | \bareta^{(j)} |\|_{2q} 
        \lesssim_q 
        \sqrt{\frac{\ttrace(\Gamma)}{N}}
        =\sqrt{|\Gamma| \frac{r_2(\Gamma)}{N}},$ and by Lemma~\ref{lem:KalmanOpCtyBdd} and the forecast covariance bound established in Subsection \ref{sssec:forecastinduction}, we have 
    \begin{align*}
        \| | \msK(\hatCpr^{(j)})|  \|_{2q} 
        &\le |H| |\Gamma^{-1}| \| |\hatCpr^{(j)} |\|_{2q}
        \le |H| |\Gamma^{-1}| \inparen {
        \| |\hatCpr^{(j)} -\Cpr^{(j)}|\|_{2q} + |\Cpr^{(j)}|
        }\\
        &\le
         |H| |\Gamma^{-1}| 
        c
        \inparen{
        1\lor  \Omega}.
    \end{align*}
    Therefore, $ \| | \msK(\hatCpr^{(j)}) \bareta^{(j)}|\|_q \le c \Omega$.

    \subsubsection{Covariance Bound} By Lemma~\ref{lem:CovarOpCtyBddLp}, we have 
    \begin{align*}
        \normn{ | \hatCpost^{(j)} - \Cpost^{(j)}| }_{q}
        &\le 
        \normn{ | \msC(\hatCpr^{(j)}) - \msC(\Cpr^{(j)})| }_{q}
        + \normn{|\hatO^{(j)}|}_q\\
        &\le
        \normn{| \hatCpr^{(j)} - \Cpr^{(j)}|}_{q}
        (1 + |A|^2 |\Gamma^{-1}| |\Cpr^{(j)}|)\\
        &+ (
            |A|^2 |\Gamma^{-1}| + 
            |A|^4 |\Gamma^{-1}|^2 |\Cpr^{(j)}|
        ) 
        \normn{|\hatCpr^{(j)}|}_{2q} \normn{|\hatCpr^{(j)}-\Cpr^{(j)}|}_{2q}
        + \normn{|\hatO^{(j)}|}_q. 
    \end{align*}
    The induction step for the forecast covariance has been proved in Subsection \ref{sssec:forecastinduction}, and so in order to show the induction step for the covariance bound we only need to control the offset term. First, using 
    the triangle inequality,  we write
    \begin{align*}
        \normn{|\whatO^{(j)}|}_q &\le  
        \normn{ |\msK(\hatCpr^{(j)}) (\widehat{\Gamma}^{(j)} - \Gamma) \msK^\top(\hatCpr^{(j)})| }_q 
        +\normn{| \bigl(I-\msK(\hatCpr^{(j)}) H\bigr) \whatC^{(j)}_{\, u \eta} \msK^\top(\hatCpr^{(j)})| }_q \\
        &+\normn{ | \msK(\hatCpr^{(j)}) (\whatC^{(j)}_{\, u \eta})^\top \bigl(I-H^\top\msK^\top(\hatCpr^{(j)})\bigr)| }_q
        = 
        \| \hatO^{(j)}_1\|_q 
        +
        \| \hatO^{(j)}_2\|_q
        +
        \| \hatO^{(j)}_3\|_q.
    \end{align*}
We next bound each term in turn.
        \paragraph{Controlling $\| \hatO^{(j)}_1\|_q$}
        By Lemma~\ref{lem:KalmanOpCtyBdd}, the forecast covariance bound established in Subsection \ref{sssec:forecastinduction}, and Theorem~\ref{thm:CovarianceExpectationBoundLq}, it holds that 
         \begin{align*}
             \| \hatO^{(j)}_1\|_q  
             &=\normn{ | \msK(\hatCpr^{(j)}) (\widehat{\Gamma}^{(j)} - \Gamma) \msK^\top(\hatCpr^{(j)})| }_q\\
             &\le
             \normn{ |\msK(\hatCpr^{(j)})| }^2_{4q} \normn{| \widehat{\Gamma}^{(j)} - \Gamma| }_{2q}
             \le
             |H|^2 |\Gamma^{-1}|^2 \normn{ |\hatCpr^{(j)}| }^2_{4q} \normn{| \widehat{\Gamma}^{(j)} - \Gamma| }_{2q}\\
             &\le
             c 
            \inparen{1 \lor\Omega}^2
            \sqrt{\frac{r_2(\Gamma)}{N}}
            \le
             c \sqrt{\frac{r_2(\Gamma)}{N}},
         \end{align*}
         where the last inequality uses that by assumption $N \ge r_2(\Cpost^{(0)}) \lor r_2(\xiCovariance) \lor r_2(\Gamma)$.
                
        \paragraph{Controlling $\| \hatO^{(j)}_2\|_q$}  By Lemma~\ref{lem:KalmanOpCtyBdd}, inequality \eqref{eq:CrossCovarianceBoundInductiveStep}, and the forecast covariance bound established in Subsection \ref{sssec:forecastinduction}, we get 
         \begin{align*}
             \| \hatO^{(j)}_2\|_q   
             &= \normn{ 
             | \bigl(I-\msK(\hatCpr^{(j)}) H\bigr) \whatC^{(j)}_{\, u \eta} \msK^\top(\hatCpr^{(j)})|
             }_q
             \le 
            \normn{
            |\msK(\hatCpr^{(j)})| 
            |I-\msK(\hatCpr^{(j)}) H| 
            |\whatC^{(j)}_{\, u \eta}| }_q \\
             &\le 
             \normn{
             |\msK(\hatCpr^{(j)})| 
             (1+|\msK(\hatCpr^{(j)})| |H|)  
             |\whatC^{(j)}_{\, u \eta}| }_q\\
             &\le 
             \normn{|\whatC^{(j)}_{\, u \eta}| }_{2q}
             \inparen{
             |H| |\Gamma^{-1}|  \| |\hatCpr^{(j)}|\|_{2q}
             + |H|^3 |\Gamma^{-1}|^2  \| |\hatCpr^{(j)}|^2\|_{2q}
             } \le c\Omega.
         \end{align*}

        \paragraph{Controlling $\| \hatO^{(j)}_3\|_q$} Note that $\| \hatO^{(j)}_3\|_q=\| \hatO^{(j)}_2\|_q$.

\section{Conclusions}\label{sec:Conclusions}
This paper has investigated $\mathsf{REnKF}$, a modification of $\mathsf{EnKF}$ with improved theoretical guarantees. Theorem \ref{thm:MultiStepPOEnKFBounds} gives non-asymptotic error bounds for a stochastic $\mathsf{EnKF}$ over multiple assimilation cycles. Numerical experiments demonstrate that the benefits of introducing resampling for theory purposes do not come at the price of a deterioration in state estimation or uncertainty quantification tasks.

Resampling techniques for ensemble Kalman algorithms deserve further research. From a theory viewpoint, resampling offers a promising path to develop long-time filter accuracy theory, blending our inductive analysis with existing results that ensure long-time stability of the filtering distributions \cite{sanz2015long}. From a methodological viewpoint, other resampling schemes can be considered \cite{naesseth2018variational}. Finally, while our numerical investigation has focused on settings where the standard $\mathsf{EnKF}$ algorithm is effective, an important open problem is to identify dynamical systems and/or observation models for which resampling may offer an empirical advantage.

\section*{Acknowledgments}
The authors are grateful for the support of NSF DMS-2027056, DOE DE-SC0022232, and the BBVA Foundation.

\vspace{10mm}
\bibliographystyle{siam} 
\bibliography{references}

\appendix
\section{Metrics for Numerical Results} \label{sec:metrics} 
In this appendix, we give a more extensive description of the Monte Carlo procedure utilized to calculate the metrics referred to in Section~\ref{sec:numerics}. We summarize the approach in Algorithm~\ref{algMetrics}. We require the following additional notation: We write $\tdiag(A) = (A_{11}, A_{22},\dots, A_{dd})^\top$. For a function $g:\R \to \R$, $g(u) = \bigl(g(u(1)),\dots, g(u(d))\bigr)^\top$ is the element-wise application of $g$ to $u$.

\begin{algorithm} 
\caption{\label{algmetric} Metrics Calculation for Numerical Results}
\begin{algorithmic}[1] \label{algMetrics}

\STATE {\bf Fixed Quantities}:  
Ground-truth state  $\{u^{(j)}\}_{ j=0}^{J}$, observations $\{y^{(j)}\}_{j=1}^J$, and Kalman filter means $\{\mpost^{(j)}\}_{j=1}^J$. 
\vspace{0.25cm}

\STATE {\bf Monte Carlo Trials}:  For $m = 1, 2, \ldots, M$ run algorithm $\mathsf{A} \in \{\mathsf{EnKF}(\ref{algEnKF}), \mathsf{REnKF}(\ref{algEnKFresample}) \}$ 
and obtain 
$\{\hatmpost^{(j),\mathsf{A}}_{m}, \hatSigma_m^{(j), \mathsf{A}}\}_{j,m}^{J, M}$.\\

{\it Mean Error: } 
\begin{align}
\begin{split}
\mathsf{E}_{m, \text {\tiny Linear}}^{\mathsf{A}} &= \frac{1}{J} \sum_{j=1}^J |\hatmpost^{(j), \mathsf{A}}_m - \mu^{(j)}|_2, \qquad 
\mathsf{E}_{m, \text {\tiny L96}}^{\mathsf{A}} = \frac{1}{J} \sum_{j=1}^J |\hatmpost^{(j),\mathsf{A}}_m - u^{(j)}|_2.
\end{split}
\end{align}

{\it Confidence Interval: }  Let $\hatsigma_{m}^{(j), \mathsf{A}} =  \sqrt{\tdiag(\hatSigma_m^{(j), \mathsf{A}})}$, then compute 
\begin{equation}
    \begin{aligned}
        \mathsf{I}_{m}^{(j), \mathsf{A}} &=  
        \hatmpost_m^{(j), \mathsf{A}} \pm 1.96 \times \hatsigma_{m}^{(j), \mathsf{A}},
        \qquad 
        && \text{(Interval)}\\
        \mathsf{W}_{m}^{\mathsf{A}} 
        &=\frac{2 \times 1.96}{dJ} \sum_{j=1}^{J} |\hatsigma_{m}^{(j), \mathsf{A}}|_1,
        \qquad 
        && \text{(Average Width)}\\
        {\mathsf{V}}_m^{\mathsf{A}} 
        &= \frac{1}{dJ}\sum_{j=1}^J\sum_{i=1}^d  \indicator \{  
        u^{(j)}(i) \in \mathsf{I}_{m}^{(j), \mathsf{A}}(i) \}.
        \qquad 
        && \text{(Average Coverage)}\\
    \end{aligned}
\end{equation}

\STATE{\bf Output}:
\begin{align}
\begin{split}
\mathsf{E}_{\text {\tiny Linear}}^{\mathsf{A}} 
&= \frac{1}{M} \sum _{m=1}^M \mathsf{E}_{m, \text {\tiny Linear}}^{\mathsf{A}}, 
\qquad
\mathsf{E}_{\text {\tiny L96}}^{\mathsf{A}} 
= \frac{1}{M} \sum_{m=1}^M \mathsf{E}_{m, \text {\tiny L96}}^{\mathsf{A}},\\
\mathsf{W}^{\mathsf{A}} 
&= \frac{1}{M}\sum_{m=1}^M  \mathsf{W}_{m}^{\mathsf{A}}, 
\qquad 
\mathsf{V}^{\mathsf{A}} 
= \frac{1}{M} \sum_{m=1}^M \mathsf{V}_m^{\mathsf{A}}.
\end{split}
\end{align}

\end{algorithmic}
\end{algorithm}

\section{Technical Results}
\subsection{Additional Notation}
Given a non-decreasing, non-zero convex function $\psi:[0,\infty] \to [0, \infty]$ with $\psi(0)=0$, the Orlicz norm of a real random variable $X$ is $\normn{X}_{\psi} = \inf\{t > 0: \E [ \psi (t^{-1}|X|)] \le 1\}$. In particular, for the choice $\psi_p(x) = e^{x^p}-1$ for $p \ge 1$, real random variables that satisfy $\normn{X}_{\psi_2} < \infty$ are referred to as sub-Gaussian, and those that satisfy $\normn{X}_{\psi_1} < \infty$ are sub-Exponential. The random vector $Y$ is sub-Gaussian (sub-Exponential) if $\normn{v^\top Y}_{\psi_2} < \infty$ ($\normn{v^\top Y}_{\psi_1} < \infty$) for any vector $v$ satisfying $|v|_2=1$.
\vspace{10mm}
\subsection{Background Results}
\begin{lemma}[{Properties of the Kalman Gain Operator \cite[Lemma 4.1 \& Corollary 4.2]{kwiatkowski2015convergence}}] \label{lem:KalmanOpCtyBdd} 
Let $\msK$ be the Kalman gain operator defined in \eqref{eq:KalmanGainOperator}. Let $P, Q \in \mcS_+^d$, $\Gamma \in \mcS_{++}^k,$ and $\H \in \R^{k \times d}$. The following hold:
\begin{align*}
\begin{split}
        |\msK(Q) - \msK(P)| 
        &\le 
        |Q - P| 
        |\H| |\Gamma^{-1}|
        \Bigl(1 + 
        \min \inparen{|P|, |Q|} 
        |\H|^2 |\Gamma^{-1}| \Bigr),\\
        |\msK(Q)| &\le |Q| |\H| |\Gamma^{-1}|,\\
        |I-\msK(Q) \H| 
        &\le 1+|Q| |\H|^2 |\Gamma^{-1}|.
\end{split}
     \end{align*}
\end{lemma}

    \begin{lemma}[{Properties of the Mean-Update Operator \cite[Lemma 4.10]{kwiatkowski2015convergence}}] \label{lem:MeanOpCtyBddLp}
        Let $\msM$ be the mean-update operator defined in \eqref{eq:MeanOperator}. 
        Let $m \in \R^d$ be a random vector and $Q$ be a random matrix such that $Q \in \mcS_+^d$ almost surely. Let $P \in \mcS_+^d$, $\Gamma \in \mcS_{++}^k$, $\H \in \R^{k \times d}$, $y \in \R^k,$ and $m' \in \R^d$ be deterministic.
        Then, for any $1 \le q <\infty$ and $y \in \R^k$ it holds that
         \begin{align*}
         \begin{split}
                \normn{ |\msM(m, Q; y)- \msM(m', P; y) |}_{q} 
                &\le 
                \norm{|m-m'|}_{q}
                + 
                |H|^2 |\Gamma^{-1}| \normn{|Q|}_{2q} \norm{|m-m'|}_{2q}\\   
                & + 
                \normn{|Q-P|}_{q} |\H| |\Gamma^{-1}|
                \bigl(1 + |\H|^2 |\Gamma^{-1}| |P| \bigr) 
                |y- \H m'|.
        \end{split}        
             \end{align*}
        \end{lemma}
\begin{lemma}[{Properties of the Covariance-Update Operator
\cite[Lemmas 4.6 \& 4.8]{kwiatkowski2015convergence}}] \label{lem:CovarOpCtyBddLp} 
    Let $\msC$ be the covariance-update operator defined in \eqref{eq:CovarianceOperator}. Let $P,Q, m, m',y, H$ and $\Gamma$ all be defined as in Lemma~\ref{lem:MeanOpCtyBddLp}.
    Then, for any $1 \le q <\infty$, it holds that
      \begin{align*}
      \begin{split}
        0 &\preceq \msC(Q) \preceq Q,\qquad 
            |\msC(Q)| \le  |Q|,\\
            \normn{ | \msC(Q) - \msC(P)| }_{q} 
            &\le
            \normn{| Q - P|}_{q}
            (1 + |\H|^2 |\Gamma^{-1}| |P|)\\
            &+ (
                |\H|^2 |\Gamma^{-1}| + 
                |\H|^4 |\Gamma^{-1}|^2 |P|
            ) 
            \normn{|Q|}_{2q} \normn{|Q-P|}_{2q}.
        \end{split}    
         \end{align*}
    \end{lemma}
    
\begin{theorem}[{Gaussian Norm Concentration, \cite[Exercise 6.3.5]{vershynin2018high}}] \label{thm:SubGaussianConcentration} Let $X \in \R^d$ be a Gaussian random vector with $\E[X] = \mu^X$, $\tvar [X]=\Sigma^X$. Then, for any $t \ge 1,$ with probability at least $1-ce^{-t}$ it holds that 
    \begin{align*}
        | X- \mu^X |_2 \lesssim  
        \sqrt{\ttrace(\Sigma^X)} + \sqrt{t|\Sigma^X| } 
        \lesssim 
        \sqrt{|\Sigma^X| (r_2(\Sigma^X) \lor t )} \,.
    \end{align*}
\end{theorem}

\begin{theorem}[{Covariance Bound, \cite[Corollary 2]{koltchinskii2017concentration}}] \label{thm:CovarianceExpectationBoundLq}
    Let $X_1,\dots, X_n$ be i.i.d. copies of a $d$-dimensional Gaussian vector $X$ with $\E[X]=0$ and $\tvar[X]=\Sigma$. Let $\hatSigma = \frac{1}{n}\sum_{i=1}^n X_i X_i^\top$ be the sample covariance estimator.
    For any $q \ge 1,$ it holds that
    \begin{align*}
        \| |\hatSigma - \Sigma | \|_q \lesssim_q  
        |\Sigma| \inparen{
        \sqrt{\frac{r_2(\Sigma)}{n}}
        \lor 
        \frac{r_2(\Sigma)}{n}
        }.
    \end{align*}
\end{theorem}

\begin{lemma}\label{lem:TraceProperties}
    Let $A, B\in \mcS^d_+$. It holds that 
    \begin{align*}
        \ttrace(AB) \le |A| \ttrace(B).
    \end{align*}
\end{lemma}

\begin{lemma}[Cross-Covariance Estimation ---Unstructured Case] \label{lem:CrossCovarianceBound}
    Let $\upr_1,\dots, \upr_N \in \R^d$ be i.i.d. Gaussian random vectors with $\E[ \upr_1] = m$ and $\tvar[\upr_1] = \Cpr$. Let $\eta_1,\dots,\eta_N \in \R^k$ be i.i.d. Gaussian random vectors with $\E [\eta_1] = 0$ and $\tvar[\eta_1]=\Gamma$, and assume that the two sequences are independent. Let
    \begin{align*}
        \hatCpr_{\, \upr\eta} = \frac{1}{N-1} \sum_{n=1}^{N} 
        (\upr_n - \hatmpr)(\eta_n  - \bareta)^{\top},
    \end{align*}
    and assume that $N \ge r_2(\Cpr) \lor r_2(\Gamma)$. Then, 
    \begin{align*}
        \| |\hatCpr_{\, \upr\eta}|\|_q
        \lesssim_q
        (|\Cpr| \lor |\Gamma|) 
        \inparen{
        \sqrt{\frac{r_2(\Cpr)}{N}} 
        \lor  
        \sqrt{\frac{r_2(\Gamma) }{N}} 
        }.
    \end{align*}
\end{lemma}
\begin{proof}
    By \cite[Lemma A.3]{ghattas2022non}, there exists a constant $c$ such that, for all $t \ge 1,$ it holds with probability at least $1-ce^{-t}$ that
    \begin{align*}
        |\hatCpr_{\, \upr\eta}|
        \lesssim
        (|\Cpr| \lor |\Gamma|) 
        \inparen{
        \sqrt{\frac{r_2(\Cpr)}{N}} 
        \lor  
        \sqrt{\frac{r_2(\Gamma) }{N}} 
        \lor 
        \sqrt{\frac{t}{N}} 
        \lor  
        \frac{t}{N}
        }.
    \end{align*}
    Integrating the tail bound then yields the result.
\end{proof}

\begin{lemma} \label{lem:TraceConcentration2}
Let $X_1,\dots, X_n$ be i.i.d. copies of a $d$-dimensional Gaussian vector $X$ with $\E[X] =0$ and $\tvar[X]=\Sigma$. Let $\hatSigma = \frac{1}{n}\sum_{i=1}^n X_i X_i^\top$ be the sample covariance estimator. Then, for any $\delta \ge 1$, it holds with probability at least $1-2e^{-\delta}$ that 
\begin{align*}
    |\ttrace(\hatSigma)-\ttrace(\Sigma)| \le 
    c \ttrace(\Sigma) \inparen{\sqrt{\frac{\delta}{n}} \lor \frac{\delta}{n}}.
\end{align*}
Further, for any $q \ge 1,$
    \begin{align*}
        \normn{ |\ttrace(\hatSigma)-\ttrace(\Sigma)| }_q \lesssim_q \frac{\ttrace(\Sigma)}{\sqrt{n}}.
    \end{align*}
\end{lemma}
\begin{proof}
    Let $Z_{ij} = \Sigma^{-1/2}_{jj} X_{ij}$ and note that, for any $t>0,$
    \begin{align*}
        \P(|\ttrace(\hatSigma)-\ttrace(\Sigma)| > t)
        =
        \P(|\ttrace(\hatSigma-\Sigma)| > t)
        &=
        \P \inparen{
        \abs{\sum_{i=1}^n \inparen{\sum_{j=1}^d (X^2_{ij} - \E X^2_{ij})  }} > nt}\\
        &=
        \P \inparen{
        \abs{\sum_{i=1}^n \inparen{\sum_{j=1}^d \Sigma_{jj} (Z^2_{ij} - \E Z^2_{ij})  }} > nt}.
    \end{align*}
    Note that the random variables $\sum_{j=1}^d \Sigma_{jj} (Z^2_{ij} - \E Z^2_{ij})$ for $i=1,\dots,n$ are independent, mean-zero and sub-exponential with $\psi_1$ norm at most $C\ttrace(\Sigma),$ since  
    \begin{align*}
        \norm{\sum_{j=1}^d \Sigma_{jj} (Z^2_{ij} - \E Z^2_{ij}) }_{\psi_1}
        &\le \sum_{j=1}^d \Sigma_{jj} \norm{ Z^2_{ij} - \E Z^2_{ij}}_{\psi_1} 
        \le C \sum_{j=1}^d \Sigma_{jj} \norm{ Z^2_{ij}}_{\psi_1} \\
        &= C \sum_{j=1}^d \Sigma_{jj} \norm{Z_{ij}}^2_{\psi_2} 
       \le C  \sum_{j=1}^d \Sigma_{jj}
        = C \ttrace(\Sigma).
    \end{align*}
    The second inequality holds due to the Centering Lemma, \cite[Lemma 2.6.8]{vershynin2018high}. Therefore, by Bernstein's inequality we have 
    \begin{align*}
        \P \inparen{
        \abs{\sum_{i=1}^n \inparen{\sum_{j=1}^d \Sigma_{jj} (Z^2_{ij} - \E Z^2_{ij})  }} > nt}
        &\le 
        2 \exp \inparen{-c \min \inparen{
        \frac{nt^2}{(\ttrace(\Sigma))^2},
        \frac{nt}{\ttrace(\Sigma)}
        }}.
    \end{align*}
    For the expectation bound, we note that 
    \begin{align*}
        \normn{ |\ttrace(\hatSigma)-\ttrace(\Sigma)| }_q^q
        &= \int_0^\infty \P(|\ttrace(\hatSigma)-\ttrace(\Sigma)|^q > t) \, dt\\
        &\le \zeta^q + q \int_C^\infty t^{q-1} \P(|\ttrace(\hatSigma)-\ttrace(\Sigma)| > t) \, dt\\
        &\le \zeta^q + 2q \int_0^\infty t^{q-1} 
         \exp \inparen{-c \min \inparen{
        \frac{nt^2}{(\ttrace(\Sigma))^2},
        \frac{nt}{ \ttrace(\Sigma)}
        }}dt\\
        &= \zeta^q + 2qc 
        \max \inparen{
        \frac{\Gamma(q/2) (\ttrace(\Sigma))^q}{n^{q/2}},
        \frac{\Gamma(q) (\ttrace(\Sigma))^q}{n^{q}}
        }.
    \end{align*}
    Taking $\zeta=\ttrace(\Sigma)/n$, it then follows that 
    \begin{align*}
        \normn{ |\ttrace(\hatSigma)-\ttrace(\Sigma)| }_q
        & \lesssim 
        \zeta + c\ttrace(\Sigma) 
        \max \inparen{
        \frac{1}{\sqrt{n}},
        \frac{1}{n}
        } \lesssim \frac{ \ttrace(\Sigma)}{\sqrt{n}}.
    \end{align*}
\end{proof}

\begin{lemma} \label{lem:GaussianTraceLqBound}
    Let $X_1,\dots, X_n$ be i.i.d. copies of a $d$-dimensional Gaussian vector $X$ with $\E[X]=\mu$ and $\tvar[X]=\Sigma$. Let $\barX = \frac{1}{N}\sum_{n=1}^N X_n$. Then, for any $q \ge 1,$
    \begin{align*}
    \| | \bar{X} - \mu | \|_q
    \lesssim_q 
    \sqrt{\frac{\ttrace(\Cpost)}{N}} \,.
\end{align*}
\end{lemma}

\begin{proof}
    Let $c_2 := c_1 \sqrt{\ttrace(\Cpost)/N}$ where $c_1$ is a sufficiently large positive constant, then 
\begin{align*}
    \E \insquare{| \bar{X} - \mu |^q}
    &= \int_0^\infty \P (| \bar{X} - \mu |^q > y ) \, dy
    \le c_2^q + \int_{c_2}^\infty \P (| \bar{X} - \mu |^q > y ) \, dy\\
    &= c_2^q+ \int_{c_2}^\infty q y^{q-1} \P (| \bar{X} - \mu | > y )\, dy\\
    &= c_2^q+ \int_{c_2- c\ttrace(\Cpost/N)}^\infty q \inparen{c\sqrt{\frac{\ttrace(\Cpost) }{N}} + t}^{q-1} 
    \P \inparen{  | \bar{X} - \mu | > c\sqrt{\frac{\ttrace(\Cpost) }{N}} + t } dt
    \end{align*}
    where the last equality holds by a change of variable. By Theorem~\ref{thm:SubGaussianConcentration} it follows that $\P(|\barX-\mu| \ge c\sqrt{\ttrace(\Sigma)} + t) \le \exp(-ct^2/|\Sigma|)$, and so the expression in the above display is bounded above by
    \begin{align*}
    &c_2^q
    + \int_{c_2- c\ttrace(\Cpost/N)}^\infty q 
    \inparen{ c\sqrt{\frac{\ttrace(\Cpost) }{N}} + t}^{q-1} 
    \exp \inparen{-\frac{c_2nt^2}{|\Cpost|} } dt\\
    &\lesssim
    c_2^q
    + \int_{0}^\infty q \inparen{
    \inparen{\frac{c\ttrace(\Cpost)}{N}}^{(q-1)/2} + t^{q-1} }
    \exp \inparen{-\frac{c_2nt^2}{|\Cpost|} } dt\\
    &=
    c_2^q
    + 
    q \inparen{
    \frac{1}{2} \Gamma(q/2) \inparen{\frac{|\Cpost|}{N}}^{q/2}
    + 
    \frac{1}{2} \inparen{\frac{c\ttrace(\Cpost) }{N}}^{(q-1)/2}
     \sqrt{\frac{\pi |\Cpost|}{N}}}\\
     &\lesssim
    c_2^q
    + 
    q \inparen{
    \frac{1}{2} \Gamma(q/2) \inparen{\frac{|\Cpost|}{N}}^{q/2}
    + 
    \frac{1}{2} \inparen{\frac{c\ttrace(\Cpost) }{N}}^{q/2}}\\
    & \lesssim 
    \inparen{\frac{\ttrace(\Cpost) }{N}}^{q/2}.
\end{align*}
Therefore, 
\begin{align*}
    \normn{ | \bar{X} - \mu | }
    \lesssim_q 
    c_2 + 
    \sqrt{\frac{|\Cpost|}{N}}
    + 
    c \sqrt{\frac{\ttrace(\Cpost) }{N}}
    \lesssim 
    \sqrt{\frac{\ttrace(\Cpost) }{N}},
\end{align*}
where the last inequality holds since $\ttrace(\Cpost) \ge |\Cpost|$ and the choice of $c_2$. 
\end{proof}

\end{document}

%% file: ex_shared.tex
\usepackage{lipsum}
\usepackage{amsfonts}
\usepackage{graphicx}

\usepackage{enumitem}

\headers{Ensemble Kalman Filters with Resampling}{Omar Al-Ghattas, Jiajun Bao, and Daniel Sanz-Alonso}

\title{Ensemble Kalman Filters with Resampling}

\author{Omar Al-Ghattas\thanks{Department of Statistics, University of Chicago, Chicago, IL (\email{ghattas@uchicago.edu,sanzalonso@uchicago.edu}}).
     \and Jiajun Bao\thanks{Committee on Computational and Applied Mathematics, University of Chicago, Chicago, IL (\email{jiajunb@uchicago.edu}}).
  \and Daniel Sanz-Alonso\footnotemark[1]
 }

\newcommand{\R}{\mathbb{R}}

\newcommand{\N}{\mathbb{N}}

\newcommand{\E}{\operatorname{\mathbb{E}}}

\newcommand{\Sam}{N}
\newcommand{\sam}{n}

\renewcommand{\H}{\mathcal{H}}

\newcommand{\msK}{\mathscr{K}}

\newcommand{\Nc}{\mathcal{N}}

\newcommand{\iid}{\stackrel{\text{i.i.d.}}{\sim}}

    \usepackage{mathtools}

    \usepackage{amsmath,amssymb}

    \usepackage{mathrsfs} 

    \usepackage{bm} 

\usepackage{bbold}

\usepackage{mathtools}
\mathtoolsset{centercolon}  

\renewcommand{\phi}{\varphi}

\usepackage[]{algorithm}
\usepackage{algorithmic}
\usepackage{enumitem}

\usepackage{graphicx}


\usepackage{tikz}
\usetikzlibrary{shapes,arrows,decorations.pathmorphing,backgrounds,positioning,fit,petri}
\usetikzlibrary{calc}
\usetikzlibrary{matrix}
\usetikzlibrary{backgrounds}
\usetikzlibrary{shapes.geometric}
\usepackage{pgfplots}
\pgfplotsset{compat=newest}
\usepgfplotslibrary{groupplots}
\pgfdeclarelayer{background}
\pgfdeclarelayer{bBack}
\pgfsetlayers{bBack,background,main}
\usetikzlibrary{fit}

\newcommand{\comments}[1]{}

\newcommand{\upperRomannumeral}[1]{\uppercase\expandafter{\romannumeral#1}}

\renewcommand{\hat}{\widehat}

\newcommand{\tdiag}{\mathrm{diag}}

\newcommand{\ttrace}{\mathrm{Tr}}

\newcommand{\indicator}{{\bf{1}}}

\renewcommand{\H}{\mathbb{H}}

\renewcommand{\P}{\mathbb{P}}

\newcommand{\hatSigma}{\widehat{\Sigma}}
\renewcommand{\H}{H}

\newcommand{\mfq}{\mathfrak{q}}

\newcommand{\whatC}{\widehat{C}}

\newcommand{\whatO}{\widehat{O}}

\newcommand{\hatO}{\hat{O}}

\newcommand{\hatsigma}{\hat{\sigma}}
\newcommand{\hatGamma}{\hat{\Gamma}}

\newcommand{\overbar}[1]{\mkern 1.5mu\overline{\mkern-1.5mu#1\mkern-1.5mu}\mkern 1.5mu}

\newcommand{\barX}{\overbar{X}}

\newcommand{\bareta}{\overbar{\eta}}

\newcommand{\mcS}{\mathcal{S}}

\newcommand{\msC}{\mathscr{C}}

\newcommand{\msM}{\mathscr{M}}

\usepackage{multirow}

\newcommand{\inparen}[1]{\left(#1\right)}             
\newcommand{\insquare}[1]{\left[#1\right]}             
\newcommand{\norm}[1]{\ensuremath{\left\lVert #1 \right\rVert}}
\newcommand{\normn}[1]{\ensuremath{\lVert #1 \rVert}}

\newcommand{\abs}[1]{\ensuremath{\left\lvert #1 \right\rvert}}

\newcommand{\tvar}{\mathrm{var}}

\newcommand{\mpost}{\mu}
\newcommand{\mpr}{m}
\newcommand{\hatmpost}{\widehat{\mu}}
\newcommand{\hatmpr}{{\widehat{m}}}
\newcommand{\Cpost}{\Sigma}

\newcommand{\Cpr}{C}
\newcommand{\hatCpost}{\widehat{\Sigma}}
\newcommand{\hatCpr}{\widehat{C}}
\newcommand{\upr}{u}

\newcommand{\xiCovariance}{\Xi}

\usepackage{graphicx}

\renewcommand{\phi}{\varphi}

\newcommand{\barxi}{\bar{\xi}}
\newcommand{\hatxiCovariance}{\hat{\xiCovariance}}